\documentclass[english]{article}

\usepackage[left=2.7cm,bottom=2.7cm,right=2.7cm,top=2.7cm]{geometry}

\usepackage{color,hyperref}

\usepackage[numbers]{natbib}
\usepackage{csquotes}
\usepackage{verbatim}
\usepackage {multirow}

\usepackage{epsfig}
\usepackage{color}

\usepackage{balance}  
\usepackage{algorithmic}
\usepackage{algorithm}
\usepackage{multirow}
\usepackage{graphicx}
\usepackage{amsmath}
\usepackage{amssymb}
\usepackage{amsfonts}
\usepackage{xspace}
\usepackage{url}
\usepackage{bbm}
\usepackage{subcaption}
\usepackage{enumitem}

\newcommand{\R}{\mathbb{R}}
\newcommand{\T}{\top}
\newcommand{\la}{\langle}
\newcommand{\ra}{\rangle}

\newcommand{\lam}{\lambda}
\newcommand{\dt}{\delta}
\newcommand{\Dt}{\Delta}

\newcommand{\pa}{\partial}
\newcommand{\al}{\alpha}

\newcommand{\ga}{\gamma}

\newcommand{\na}{\nabla}
\newcommand{\lt}{\left}
\newcommand{\rt}{\right}

\newcommand{\argmax}{\mathop{{\rm argmax}}}
\newcommand{\argmin}{\mathop{{\rm argmin}}}

\newcommand{\td}{\tilde}
\newcommand{\wtd}{\widetilde}

\newcommand{\rmd}{{\rm d}}

\newcommand{\cC}{\mathcal{C}}

\newcommand{\cG}{\mathcal{G}}

\newcommand{\cS}{\mathcal{S}}

\newcommand{\cV}{\mathcal{V}}

\newtheorem{theorem}{Theorem}[section]
\newtheorem{corollary}[theorem]{Corollary}
\newtheorem{lemma}[theorem]{Lemma}
\newtheorem{proposition}[theorem]{Proposition}

\newtheorem{assumption}[theorem]{Assumption}

\makeatletter
\@addtoreset{equation}{section}
\makeatother



\begin{document}
	
	\title{Nonparametric Finite Mixture Models with Possible Shape Constraints: A Cubic
		Newton Approach}
	
	 \author{Haoyue Wang\thanks{MIT Operations Research Center (email: {haoyuew@mit.edu}).} \and
		Shibal Ibrahim\thanks{MIT Department of Electrical Engineering and Computer Science (email: {shibal@mit.edu}).}
		\and
		Rahul Mazumder\thanks{MIT Sloan School of Management, Operations Research Center and MIT Center for Statistics (email: {rahulmaz@mit.edu}).}
	}

\date{}

\maketitle

	\begin{abstract}
We explore computational aspects of maximum likelihood estimation of the mixture proportions of a nonparametric finite mixture model---a convex optimization problem with old roots in statistics and a key member of the modern data analysis toolkit. Motivated by problems in shape constrained inference, we consider structured variants of this problem with additional convex polyhedral constraints.  We propose a new cubic regularized Newton method for this problem and present novel worst-case and local computational guarantees for our algorithm. We extend earlier work by Nesterov and Polyak to the case of a self-concordant objective with polyhedral constraints, such as the ones considered herein. We propose a Frank-Wolfe method to solve the cubic regularized Newton subproblem;  and derive efficient solutions for the linear optimization oracles that may be of independent interest. In the particular case of Gaussian mixtures without shape constraints, we derive bounds on how well the finite mixture problem approximates the infinite-dimensional Kiefer-Wolfowitz maximum likelihood estimator. Experiments on synthetic and real datasets suggest that our proposed algorithms exhibit improved runtimes and scalability features over existing benchmarks.
 \end{abstract}

\section{Introduction}
In this paper, we study a problem in nonparametric density estimation---in particular, learning 
the components of a mixture model, which is a key problem in statistics and related disciplines. Consider a one-dimensional mixture density of the form: $g(x) = \sum_{i=1}^{M} w_{i} \psi_{i} (x)$ where, the $M$ components (probability densities) $\{\psi_{i}(x)\}_{1}^{M}$ are known, but the proportions $\{w_i\}_{1}^{n}$ are unknown, and need to be estimated from the data at hand.

\subsection{Learning nonparametric mixture proportions} Suppose we are given $N$ samples $\{X_{j}\}_{1}^N$, independently drawn from $g(x)$, we can obtain a maximum likelihood estimate (MLE) 
of the mixture proportions or weights via the following convex optimization problem:
\begin{equation}
\begin{aligned}\label{crit-our}
{\min_w} &~~  - \frac{1}{N}\sum_{j \in [N]} \log (\sum_{i \in [M]} w_{i} B_{ij} )\\
{\rm s.t. }~&~~ w\in \Dt_M:= \{ w~:~\sum_{i\in [M]} w_{i}=1, w_{i} \geq 0, i \in [M]\} \end{aligned}
\end{equation} 
where, $B_{ij} = \psi_{i}(X_{j}) \geq 0$ is the evaluation of the density $\psi_{i}$ at point $X_{j}$ for all $i,j$. 
One of our goals in this paper is to present new algorithms for~\eqref{crit-our} with associated computational guarantees, that scale to instances with $N\approx10^6$ and $M \approx 10^3$. Problem~\eqref{crit-our} has a rich history in statistics: this arises, for example, in the context of empirical Bayes estimation~\cite{koenker2014convex,jiang2009general,kim2020fast,brown2009nonparametric}. Several choices of the bases functions $\{\psi_{i}\}$ are possible---see for example~\cite{kim2020fast}, and references therein.
A notable example arising in Gaussian sequence models~\cite{johnstone2004needles}, is to consider $\psi_{i}(x):=\varphi(x-\mu_{i})$ where, $\varphi$ is the standard Gaussian density and the location parameter 
$\mu_{i}$ is pre-specified for all $i \in [M]$. In this case,  problem~\eqref{crit-our} can be interpreted as a finite-dimensional approximation of the original formulation of the Kiefer-Wolfowitz nonparametric MLE~\cite{kiefer1956consistency}.
Given an infinite mixture $g_{{\mathcal Q}}(x) = \int_{\R} \varphi(x-\mu) d{\mathcal Q}(\mu)$ where, ${\mathcal Q}$ is a probability distribution on $\R$, the Kiefer-Wolfowitz nonparametric MLE estimates ${\mathcal Q}$ given $N$ independent samples from $g_{{\mathcal Q}}$.
This infinite dimensional problem admits a finite dimensional solution with ${\mathcal Q}$ supported on $N$ {\emph{unknown}} atoms~\cite{kiefer1956consistency,koenker2014convex}. Since these atoms can be hard to locate, it is common to resort to discrete approximations like~\eqref{crit-our} for a suitable pre-specified grid of $\{\mu_{i}\}$-values.
In this paper, we present bounds that quantify 
how well an optimal solution to a discretized problem~\eqref{crit-our} (with a-priori specified atoms), approximates a solution to the infinite dimensional Kiefer-Wolfowitz MLE problem.

As noted by~\cite{kim2020fast}, while versions of problem~\eqref{crit-our} originated in the 1890s, the utility of this estimator and the abundance of large-scale datasets have necessitated solving problem~\eqref{crit-our} at scale (e.g., with $N\approx 10^5$-$10^6$ and $M\approx 10^3$).  
A popular algorithm for~\eqref{crit-our} is based on the Expectation Maximization (EM) algorithm~\cite[see for example]{jiang2009general}, which is known to have slow convergence behavior in practice~\cite{koenker2014convex,kim2020fast}. \cite{koenker2014convex} observed that solving a dual of~\eqref{crit-our} via Mosek's interior point solver, led to improvements over EM-based methods. Recently,~\cite{kim2020fast} propose an interesting approach based on sequential quadratic programming for~\eqref{crit-our} that offers notable improvements over~\cite{koenker2014convex,koenker2017rebayes}. \cite{kim2020fast} also use low-rank approximations of $B$, active set updates, among other clever tricks to obtain computational improvements.
Global complexity guarantees of these methods appear to be unknown. 
In this paper, we present a novel computational framework (with global complexity guarantees) that applies to Problem~\eqref{crit-our}; and a constrained variant of this problem, arising in the context of shape constrained density estimation~\cite{groeneboom2014nonparametric}, as we discuss next.

\subsection{Learning nonparametric mixtures with shape constraints}\label{sec: shape-constraints}
In many applications, practitioners have prior knowledge about the shape of the underlying density. In these cases, it is desirable to incorporate such qualitative shape constraints into the density estimation procedure---this falls under the purview of {\emph{shape constrained inference}}\cite{groeneboom2014nonparametric} which dates back to at least 1960s with the seminal work of~\cite{grenander1956theory}. 
Due to their usefulness and applicability in a wide variety of domains, the last 10 years or so has witnessed a significant interest in shape constrained inference in terms of methodology development---see~\cite{samworth2018special} for a recent overview.

We consider a structured variant of~\eqref{crit-our} which allows for density estimation under shape constraints.
That is, we assume the density $x \mapsto \sum_{i \in [M]} w_{i}\psi_{i}(x)$ obeys a pre-specified shape constraint such as monotonicity (e.g, decreasing), convexity, unimodality, among others~\cite{groeneboom2014nonparametric}.  
To this end, we consider a special family of basis functions, the Bernstein polynomial bases~\cite{carnicer1993shape}, \cite[Ch 7]{phillips2003interpolation},  which are well-known for their (a) appealing shape-preserving properties, and 
(b) ability to approximate any smooth function over a compact support to high accuracy. Density estimation with Bernstein polynomials have been studied in the statistics community~\cite{vitale1975bernstein,babu2002application,ghosal2001convergence,turnbull2014unimodal}, though the exact framework and algorithms we explore herein appear to be novel, to our knowledge. 
Given an interval $[0,1]$ (say)\footnote{For a general compact interval $[a,b] \subset \R$, one can obtain a suitable bases representation via a standard affine transformation mapping $[a,b]$ to $[0,1]$.}, the Bernstein polynomials of degree $M$ are given by a linear combination of the following Bernstein basis elements 
\begin{equation}\label{def: Bernstein kernel}
\td b_m (x) = \begin{cases} 	
\frac{\Gamma (M+1)}{\Gamma(m) \Gamma(M-m+1) } x^{m-1} (1-x)^{M-m} & \text{ if } x \in [0,1]\\
0& \text{ otherwise}
\end{cases}
\end{equation}
for $ m\in [M]$; and $\Gamma(\cdot)$ is the standard Gamma function. 
This leads to a family of 1D \emph{smooth} densities on the unit interval, of the form $g_{B}(x) = \sum_{m=1}^{M} w_{m} \td b_m (x)$ where the mixture proportions (given by $w$) lie on the $M$-dimensional simplex (i.e., $\Delta_{M}$). Note that $\td b_{m}(x)$ is a Beta-density with shape parameters $m$ and $M-m+1$; and
$g_{B}(x)$ can be interpreted as a finite mixture of Beta-densities. 

\smallskip

\noindent{\bf Estimation under polyhedral shape constraints.} 
An appealing property of the Bernstein bases representation is their shape preserving nature~\cite{carnicer1993shape,phillips2003interpolation}: a shape constraint on $m \mapsto w_{m}$ translates to a shape constraint on $x \mapsto g_{B}(x)$, restricted to its support $[0,1]$. In particular, consider the following examples:

\begin{itemize}
    \item If $w_{m} \le w_{m+1}, m\in [M-1]$ i.e., $m \mapsto w_{m}$ is an increasing sequence, then $x \mapsto g_B(x)$ is increasing. Similarly, if $m \mapsto w_{m}$ is decreasing, then $x \mapsto g_B(x)$ is a decreasing density. 

\item If $ 2w_m \ge w_{m-1} + w_{m+1}  $ for all $2\le m\le M-1$, i.e., $m \mapsto w_{m}$ is concave, then $x \mapsto g_B(x)$ is concave. Similarly, 
if $m \mapsto w_{m}$ is convex, then $x \mapsto g_B(x)$ is convex. 

\end{itemize}

The shape preserving property also extends to combinations of constraints mentioned above: for example, when the sequence $m \mapsto w_{m}$ is decreasing and convex, then so is the density $x\mapsto g_{B}(x)$. 
For all these shape constraints, $w$ lies in a convex polyhedral set.

In order to estimate the mixture density $g_{B}(x)$, we estimate the weights $\{w_{m}\}_{1}^{M}$ such that they lie in the $M$-dimensional unit simplex and in addition, obey suitable shape constraints. The MLE for this setup is given by a constrained version of~\eqref{crit-our}:
\begin{equation}\label{problem1}
\begin{aligned}
\min_w ~~ f(w) := - \frac{1}{N} \sum_{j=1}^N \log (\sum_{i \in [M]} w_{i} B_{ij})~~{\rm s.t.}~~~w\in \cC:= \Dt_M \cap S.
\end{aligned}
\end{equation}
 where, $B_{ij}=\psi_{i}(X_{j})=\tilde{b}_{i}(X_{j})$; and 
 $S$ is a convex polyhedron corresponding to the shape constraint\footnote{For example, $S = \{ w\in \R^{M}: w_{1} \le \cdots \le w_M  \}$ implies that $x\mapsto g_{B}(x)$ is decreasing. See Section~\ref{section: using FW to solve  sub-problem} for details.}.
 In addition to the shape constraints mentioned earlier,
 another useful shape restricted family arises when $m \mapsto w_{m}$ is a unimodal sequence~\cite{carnicer1993shape,turnbull2014unimodal} i.e., a sequence with one mode or maximum\footnote{In this paper, for a unimodal sequence $\{w_{m}\}_{1}^M$, we use the convention that the sequence is first increasing and then decreasing, with a maximum located somewhere in $\{1, \ldots, M\}$.}. A unimodal sequence $\{w_{m}\}$ leads to a unimodal density $g_{B}(x)$~\cite{carnicer1993shape}. In order to obtain the maximum likelihood estimate of $g_B(x)$ under the assumption that it is unimodal (but the mode location is unknown), we will need to minimize
 $f(w)$ over all unimodal sequences $w$, lying on the unit simplex. 
 If we know that the unimodal sequence $m \mapsto w_{m}$  has a mode at $m_0 \in [M]$, then $w$ lies in a convex polyhedral set $\cC_{m_0}$ (say)---this falls into the framework of~\eqref{problem1}. To find the best modal location, we will need to find the location $m_0 \in [M]$ that leads to the largest log-likelihood---this can be done by a one-dimensional search over $m_0$ (for further details, see Section~\ref{sec:unimodal-update}). 
A detailed discussion of the different polyhedral shape constraints we consider in this paper can be found in Section~\ref{sec:solving-LP-oracle}. 

The presence of additional (shape) constraints in~\eqref{problem1} poses further algorithmic/scalability challenges, compared to the special case~\eqref{crit-our}, discussed earlier. Our main goal in this paper is to present a new scalable computational framework for~\eqref{problem1} with computational guarantees.  

For motivation, Figure \ref{fig:Synthetic-compare1-shape} illustrates density estimates available from solutions of~\eqref{problem1}. We consider a synthetic dataset where the underlying density is decreasing; and estimate the density by solving~\eqref{problem1} for two choices of $S$: (a) Here we consider $S=\R^M$ (with no shape constraint), and (b) $S$ corresponds to all decreasing sequences $m \mapsto w_{m}$. Figure~\ref{fig:Synthetic-compare1-shape} shows that without shape constraints, the estimated density can have spurious wiggly artifacts, which disappear under shape constraints. The estimated weights $w$, are shown in the right panel. 
\begin{figure}[h!]
        \centering
\scalebox{0.9}{\begin{tabular}{cc}
density estimates & estimated weights \\
            \includegraphics[trim=15 5 20 35,clip,width=0.5\textwidth]{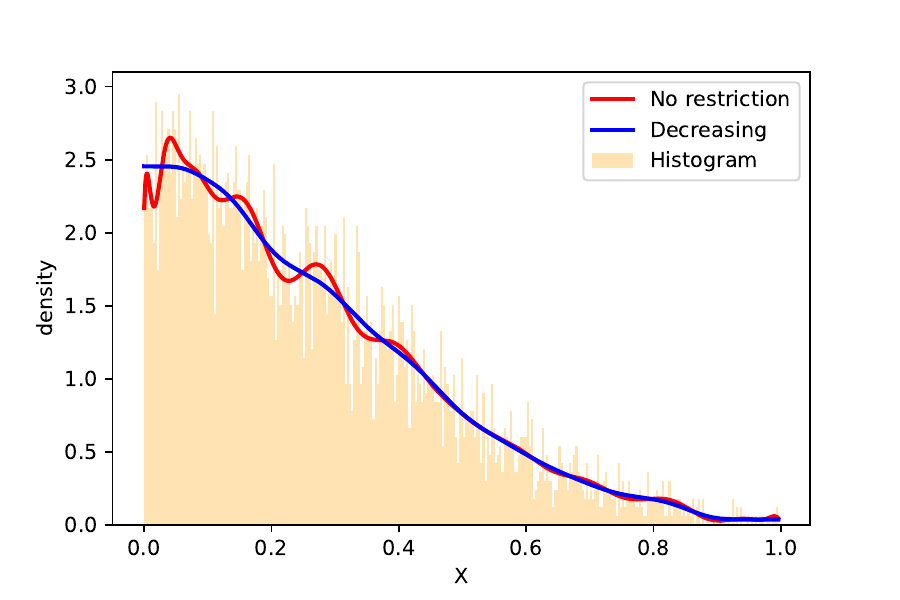}&
        \includegraphics[trim=12 5 20  35,clip,width=0.5\textwidth]{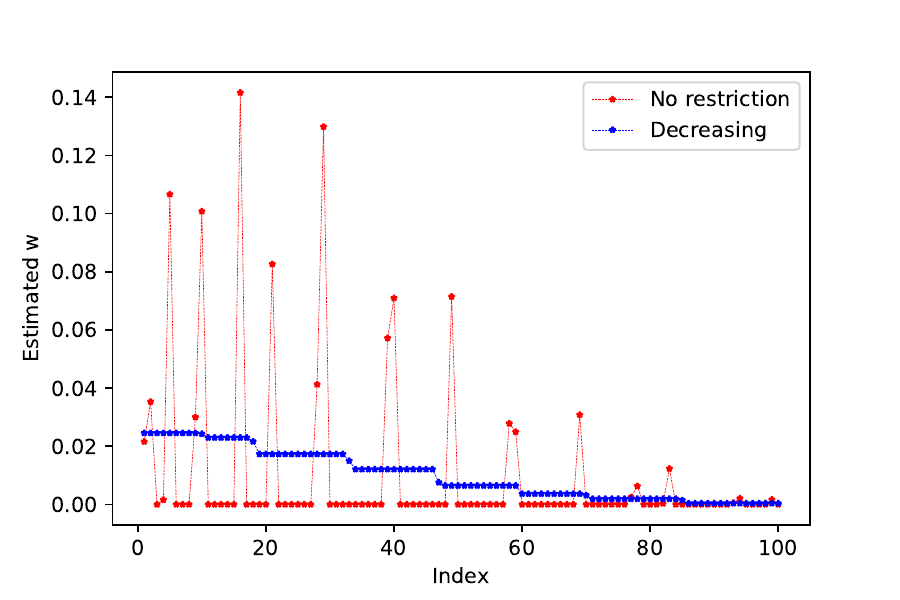}
\end{tabular}}
\caption{{\small [Left panel] Nonparametric density estimates obtained from~\eqref{problem1} under two different settings: (a) no shape restriction; (b) a decreasing constraint. 
We draw 
$N=5000$ independent samples, each $X_i$ was taken to be the absolute value of a standard Gaussian (values $\geq 3$ were discarded)---all $X_{i}$-values were then re-scaled to $[0,1]$. 
We consider $M=100$ Bernstein bases elements for both estimators. [Right panel] shows the estimated weight vector $w$ as available from~\eqref{problem1}. }}        
\label{fig:Synthetic-compare1-shape}
\end{figure}

Earlier work in statistics has studied density estimation under various shape constraints: examples include,
convex~\cite{groeneboom2001estimation}, 
monotone~\cite{groeneboom1985estimating}, 
convex and decreasing~\cite{groeneboom2001estimation} shape constraints\footnote{Another important family of shape constrained densities are log-concave densities~\cite{samworth2018recent}. While a log-concavity constraint on the weights $\{w_{i}\}$ leads to a log-concave density $x \mapsto g_{B}(x)$~\cite[see for example]{mu2015log}, a log-concavity constraint on $w_{i}$'s is not convex; and hence does not fall into our framework. This being said, log-concave densities are unimodal, and our framework can address unimodal densities on a compact support.}.
As an example of a concave density, we can consider the triangular distribution on $[0,1]$---Section~\ref{section: Experiments} presents our numerical experience with shape-constrained densities arising in real-world datasets.
The approach we pursue here, based on learning the proportions of a 
nonparametric mixture of Bernstein polynomials, differs from these earlier approaches---the density estimate available from~\eqref{problem1}, satisfies the imposed shape constraint and is also smooth.
On the other hand, nonparametric density estimation under shape constraints such as monotonicity, convexity (for example) do not lead to smooth estimators.

\subsection{Contributions} In this paper we study a new algorithmic framework for nonparametric density estimation with possible shape constraints.
Our approach is based on a maximum likelihood framework to learn the proportions of a finite mixture model subject to polyhedral constraints; and is given by Problem~\eqref{problem1}. 
To our knowledge, the general form of estimator~\eqref{problem1} that we consider here, has not been studied earlier---though, special instances of this framework has been studied before~\cite[see for example]{kim2020fast,koenker2014convex,turnbull2014unimodal,jiang2009general}. 
In the special case of Problem~\eqref{crit-our} when the bases elements correspond to a Gaussian location mixture family, we present guarantees on how well a solution to the finite mixture problem approximates the infinite version of Kiefer-Wolfowitz nonparametric MLE~\cite{kiefer1956consistency}.

Our algorithm for~\eqref{problem1} is based on a Newton method with cubic regularization, a related method was first proposed by~\cite{nesterov2006cubic} for a different family of optimization problems. Specifically, the setup of~\cite{nesterov2006cubic} does not apply to our setting as we have a self-concordant objective function $f$ with (possibly) unbounded gradients, and polyhedral constraints. Our proposed algorithm at every iteration solves a modified Newton step: this is a second order local model for $f$  with an additional cubic regularization that depends upon the local geometry of $f$ (see Section~\ref{section: Self-concordant Cubic-regularized Newton method} for details). 
We establish a worst case $O(1/k^2)$ convergence rate of our algorithm (where, $k$ denotes the number of cubic regularized Newton steps); and also derive a local quadratic convergence guarantee. 
For every cubic Newton step, we need to solve a structured convex problem for which we use Frank-Wolfe methods~\cite{frank1956an,Jaggi2015}. We present nearly closed-form solutions for the linear optimization oracles arising in the context of the Frank-Wolfe algorithms. To our knowledge, these linear optimization oracles have not been studied before, and may be of independent interest. 

We numerically compare our algorithms for~\eqref{problem1} with different choices of $\cC$, versus various existing benchmarks for both synthetic and real datasets.
For the special case of~\eqref{crit-our}, our approach leads up to a 3X-5X improvement in runtime over the recent solver~\cite{kim2020fast}; and a 10X-improvement over the commercial solver Mosek (both in terms of obtaining a moderate-accuracy solution). Furthermore, in the presence of shape constraints, our approach can lead to a 20X-improvement in runtimes over Mosek, and is more memory-friendly (especially, for larger problem instances).

\subsection{Related Work} 
Damped Newton methods~\cite{nesterov1994interior} are commonly used for unconstrained problems with a self-concordant objective. 
\cite{nesterov2013introductory,nesterov1994interior} present computational guarantees of damped Newton for unconstrained minimization of self-concordant functions. 
\cite{tran2015composite} generalize such damped Newton methods for optimization over additional (simple) constraints; and establish computational guarantees---these methods have been further generalized by~\cite{sun2019generalized}, motivated by applications in statistics and machine learning.

In a different line of work, \cite{nesterov2006cubic} introduced a cubic regularized Newton algorithm in the context of a convex optimization problem of the form $\min_{x \in \R^d} H(x)$ where, the convex function $H$ has a Lipschitz continuous Hessian:  $\|\na^2 H(x) - \na^2 H(y) \|\leq L \| x-y\|$ for all $x, y$ (here, $\| \cdot\|$ is the Euclidean norm), for some constant $L>0$. 
As noted by~\cite{nesterov2006cubic}, a unique aspect of the cubic Newton method is that it leads to
a \emph{global worst case} convergence guarantee of $O(1/k^2)$, similar guarantees are not generally available for Newton-type methods with few exceptions.
\cite{nesterov2006cubic} also show local quadratic convergence guarantees for the cubic Newton method.  
 However, as mentioned earlier, these results do not directly carry over to the setting of~\eqref{problem1} as our objective function $f(\cdot)$ does not have a Lipschitz continuous Hessian; additionally, Problem~\eqref{problem1} is a constrained optimization problem. A main technical contribution of our work is to address these challenges.

In other work pertaining to first order methods,~\cite{dvurechensky2020self} presents an analysis of the Frank Wolfe method for self-concordant objective functions with a possible constraint (e.g, a simplex or the $\ell_{1}$-ball constraint).
Proximal gradient methods~\cite{tran2015composite} have also been proposed for a similar family of problems. 
During the preparation of this manuscript, we became aware of a recent work~\cite{liu2020newton} who also consider the minimization of a family of self-concordant functions with possible constraints. They propose a proximal damped Newton-method where the resulting sub-problem is solved via Frank Wolfe methods, leading to a double-loop algorithm. 
\cite{liu2020newton} study problems in portfolio optimization, the D-optimal design and logistic regression with an Elastic net penalty---these problems involve simple constraints (e.g., $\ell_1$-ball or a simplex constraint) with well-known linear optimization oracles.  Our work differs in that we propose and study a variant of the cubic regularized Newton method which is a different algorithm with global computational guarantees in the spirit of~\cite{nesterov2006cubic}. Moreover, our focus is on nonparametric density estimation with possible shape constraints. 
In order to solve the cubic Newton step via Frank Wolfe, we derive closed-form solutions of the linear optimization oracles over different polyhedral constraints appearing in~\eqref{problem1}. To our knowledge, these linear optimization oracles have not been studied earlier.

As mentioned earlier, recently~\cite{kim2020fast} presents a specialized algorithm for~\eqref{crit-our} leading to improvements over Mosek's interior point algorithms~\cite{koenker2014convex,koenker2017rebayes} and the 
EM algorithm~\cite{jiang2009general}. The method of~\cite{kim2020fast} however does not apply to the shape constrained problem~\eqref{problem1}. \cite{kim2020fast} do not discuss complexity guarantees for their proposed approach.

\smallskip

\noindent \textbf{Organization.} The rest of the paper is organized as follows. 
Section~\ref{section: Self-concordant Cubic-regularized Newton method} presents our main algorithm, along with its global and local computational guarantees. 
In Section~\ref{section: using FW to solve  sub-problem}, we discuss how to solve the cubic-regularized Newton step via Frank-Wolfe along with the linear optimization oracles. In Section~\ref{sec: approximation-guarantees}, we derive bounds on how well a solution to the finite mixture problem approximates the infinite dimensional Kiefer-Wolfowitz problem. 
Section~\ref{section: Experiments}  presents numerical results for our proposed method and comparisons with several benchmarks. Proofs of all results and additional technical details can be found in the Supplement.

\section{Notation and Preliminaries}
For an integer $n\ge1$, let $[n]:=\{1,2,....,n\}$. 
We let $\R^n_{+} := \{x\in \R^n: x_i \ge 0 ~ \forall i\in [n]\}$ denote the nonnegative orthant; $\R^n_{++} := \{x\in \R^n: x_i > 0 ~ \forall i\in [n]\}$ denote the interior of  $\R^n_{+}$; and
$1_n$ denote the vector in $\R^n$ with all coordinates being $1$. 
For any $M \times M$ positive semidefinite matrix $A$, we define 
$\| w\|_{A} = \sqrt{w^\T A w}$ for any vector $w \in \R^M$; $A[w,u] := u^\T A w$ for any vectors $w,u$; and write $A[u]^2 := A[u,u]$. Moreover, for any symmetric tensor $T \in \R^{M\times M\times M}$ and vectors $u,v,w\in \R^M$, we define $T[u,v,w]:= \sum_{i,j,k = 1}^M T_{ijk} u_iv_jw_k$ and use the shorthand 
$T[u]^3 := T[u,u,u]$. 
For any function $g$ that is thrice continuously differentiable on an open set $U \in \R^M$, and $w \in U$, let $\na^3 f(w) \in \R^{M\times M\times M}$ be the tensor of third-order derivatives: 
$$ [\na^3 f(w)]_{i,j,k} = \frac{\pa^3}{\pa w_i \pa w_j \pa w_k} f(w), ~~~~ i,j,k \in [M].$$
For any two symmetric matrices $A,\tilde{A}\in \R^{n\times n}$, we use the notation $A\preceq \tilde{A}$ if $\tilde{A}-A$ is positive semidefinite, and $A\succeq \tilde{A}$ if $A-\tilde{A}$ is positive semidefinite. 
For matrix $B = ((B_{ij})) \in \R^{M\times N}$ in \eqref{problem1}, let the $j$-th column be denoted by $B_{j}:= [B_{1j}, B_{2j}, ..., B_{Mj}]^\T \in \R^M$ for $j \in [N]$.

\subsection{Preliminaries on self-concordant functions}
We present some basic properties of self-concordant functions following~\cite{nesterov1994interior} that we use. For further details, see~\cite{nesterov1994interior}.

Let $Q$ be an open nonempty convex subset of $\R^n$. 
A thrice continuously differentiable function $F: Q \rightarrow \R$ is called $\alpha$-self-concordant if it
is convex and for any $x\in Q$, $u\in \R^n$, it holds
\begin{equation}
|\na^3 F(x) [u,u,u]| \le 2 \alpha^{-1/2} (\na^2 F(x) [u,u])^{3/2}. \nonumber
\end{equation}
For the same function $F$ and $x\in Q$, let us define:
\begin{equation}\label{defn-norm-F,x}
\| u \|_{F,x} := (   (1/\alpha) \na^2 F(x)[u,u]    )^{1/2}. 
\end{equation}
For any $y\in Q$ such that $\bar{r}:=\| x-y \|_{F,x} <1$, we have the following ordering: 
\begin{equation}\label{eqn: hessian compare}
    (1-\bar{r})^2 \na^2 F(x) \preceq \na^2 F(y) \preceq \frac{1}{(1-\bar{r})^2} \na^2 F(x).
\end{equation}
If a function $F$ is $1$-self-concordant on $Q$, then it holds
\begin{eqnarray}\label{eqn: self-concordant strong convex}
F(y) ~\ge~ F(x) + \na F(x)^\T (y-x) + \rho \lt(\| y-x \|_{\na^2 F(x)}\rt) \ , 
\end{eqnarray}
where the function $\rho:\R_+ \rightarrow \R_+$ is given by $\rho(t) = t - \log(1+t)$.

The following result~\cite[Section 4.1.3]{nesterov2013introductory}
follows from the definition of self-concordance.
\begin{lemma}\label{lemma: f-self-concordant}
The function	$w\mapsto f(w)$ defined in \eqref{problem1} is  $(1/N)$-self-concordant. 
\end{lemma}

\section{A cubic-regularized Newton method for Problem~\eqref{problem1}}\label{section: Self-concordant Cubic-regularized Newton method}
In this section, we develop our proposed algorithmic framework for problem \eqref{problem1}. For $w,y\in \cC$, define
\begin{equation}\label{def: Phi_f}
\Phi_f(y,w) := f(w) + \na f(w)^\T (y-w) + \frac{1}{2} \na^2 f (w) [y-w]^2   \end{equation}
as a local second order model for $y \mapsto f(y)$ around the point $w$.

Our proposed method is presented in Algorithm \ref{algorithm: Self-concordant Cubic-regularized Newton Method}, which is a modification of the original cubic regularized Newton method by~\cite{nesterov2006cubic} proposed for a different family of problems. 

\begin{algorithm}[h!]
	\caption{A cubic regularized Newton method for Problem~\eqref{problem1}}
	\label{algorithm: Self-concordant Cubic-regularized Newton Method}
	\begin{algorithmic}
		\STATE \textbf{Input}: initial point $w^0 \in \cC$; initial regularization parameter $L_0>0$; two non-negative sequences $\{\ga_i\}_{i=1}^\infty$ and $\{\rho_i\}_{i=1}^\infty$; 
		enlarging parameter $\beta>1$.

		\textbf{For} $k=1,2,...$
		\STATE 1. Set $L_k = L_{k-1}$.
		
		\STATE 2. Find an approximate solution $y^k$ of
		\begin{eqnarray}\label{update-yk}
		\min_{y\in\cC}  ~\Big\{ h_f (y, w^k) := \Phi_f(y,w^k) + \frac{L_k}{6} 
		\|  y -w^k \|_{\na^2 f(w^k)}^{3}
		\Big\} \ .
		\end{eqnarray}

		\STATE 3. \textbf{If} 
		\begin{eqnarray}\label{condition1}
		f(y^k) \le \Phi_f(y^k,w^k) + (L_k/6) \|  y^k -w^k \|_{\na^2 f(w^k)}^{3} + \ga_k \ ,
		\end{eqnarray}
	~~	let $w^{k+1}$ be a point satisfying 
		\begin{eqnarray}\label{condition2}
		f(w^{k+1}) \le f(w^k) ~~ \text{and} ~~
		f(w^{k+1}) \le f(y^{k}) + \rho_k  \ .
		\end{eqnarray}
		~~
		\textbf{Else}, 
		let $L_k = \beta L_k$ and go back to step 2.

	\end{algorithmic}
\end{algorithm}

Note that Step~2 in Algorithm~\ref{algorithm: Self-concordant Cubic-regularized Newton Method} differs from that proposed in~\cite{nesterov2006cubic}---the cubic regularization term we use
$(L_k/6) \lt( \na^2 f(w^k) [y-w^k] \rt)^{3/2}$, depends upon the \emph{local curvature} of the objective function; and a local $L_{k}$.
In contrast, \cite{nesterov2006cubic} uses a cubic regularization term $L/6\| y - w^k\|^2$
(note the use of the Euclidean norm) 
and a fixed value of the Lipschitz constant (of the Hessian), $L$.
Intuitively speaking, the  regularization term in~\eqref{update-yk} restricts the distance of the update $y^k$ from $w^k$, by adapting to the  
local geometry of the self-concordant function $f(\cdot)$.
If the regularization term is too small, and the condition \eqref{condition1} in step 3 does not hold, then we increase $L_k$. In Section~\ref{subsection: Convergence guarantees}, we prove that $\sup_{k\ge 1} \{ L_k \}$ will be bounded by a constant multiple of $N$ (this does not depend upon the data $B$ or the constraints), implying that $L_k$ can be increased for at most a finite number of times.

Note that Algorithm~\ref{algorithm: Self-concordant Cubic-regularized Newton Method} allows for an inexact solution of sub-problem~\eqref{update-yk}. As we discuss in Section~\ref{section: using FW to solve  sub-problem}, this  sub-problem cannot be solved in closed form---we use a Frank Wolfe method to approximately solve the sub-problems efficiently.  
As long as the accuracy of solving \eqref{update-yk} suitably decreases as the algorithm proceeds, the global convergence still holds (See Section~\ref{subsection: Convergence guarantees}). 
 
Finally, Algorithm~\ref{algorithm: Self-concordant Cubic-regularized Newton Method} allows for a flexible choice of $w^{k+1}$ as long as it satisfies condition~\eqref{condition2} with a sequence of tolerances $\{\rho_i\}_{i=1}^\infty$. A natural special case is $w^{k+1} = y^k$ with $\rho_k = 0$ --- In our numerical experiments, we observe that this flexibility can be helpful in the initial stages of the algorithm and can lead to improved computational performance and an overall reduction of running times.

Finally, we note that Algorithm \ref{algorithm: Self-concordant Cubic-regularized Newton Method} is different from both the damped Newton method in \cite{nesterov1994interior} and the cubic-regularized Newton method in \cite{nesterov2006cubic}.

\subsection{Computational guarantees}\label{subsection: Convergence guarantees}
Here we establish the convergence properties of Algorithm~\ref{algorithm: Self-concordant Cubic-regularized Newton Method}. 
In Section~\ref{global-sublinear-CN} we establish a global (worst case) sublinear  $O(1/k^2)$ rate of convergence for Algorithm~\ref{algorithm: Self-concordant Cubic-regularized Newton Method} (See Theorem~\ref{theorem: global convergence}). An ingredient in this proof is that the adaptive parameter $L_{k}$ is bounded above by a constant that does not depend upon the problem data $B$ or the constraints (cf Lemma~\ref{lemma: Boundedness of L_k}). 
In Section~\ref{local-quadratic-CN} we establish a local quadratic convergence of
Algorithm~\ref{algorithm: Self-concordant Cubic-regularized Newton Method} in terms of $\| w^{k} - w^*\|_{H_*}$ where, $w^*$ is an optimal solution of Problem~\eqref{problem1} and $H_*=\nabla^2 f(w^*)$ is the Hessian of the objective at an optimal solution. Proofs of results in this section appear in Section~\ref{sec:append:computational-guarantees}.

\subsubsection{A global sublinear convergence rate}\label{global-sublinear-CN}
We make a mild assumption pertaining to the updates~\eqref{update-yk} in Algorithm \ref{algorithm: Self-concordant Cubic-regularized Newton Method}.
\begin{assumption}\label{assumption:  sub-problem solution}
	Suppose $y^k$ (computed in Step 2 of Algorithm \ref{algorithm: Self-concordant Cubic-regularized Newton Method}) satisfies
	\begin{equation}\label{condition3}
	\Phi_f(y^k,w^k) + ({L_k}/{6}) 
	\|  y^k -w^k \|_{\na^2 f(w^k)}^{3} \le f(w^k).
	\end{equation}
\end{assumption}
Note that condition \eqref{condition3} is satisfied if we set $y^k = w^k$, hence Assumption \ref{assumption:  sub-problem solution} should be satisfied by any reasonable solution of  sub-problem~\eqref{update-yk}. The following lemma makes use of Assumption~\ref{assumption:  sub-problem solution} to establish that the  sequence $\{L_{k}\}$ is bounded. 

\begin{lemma}\label{lemma: Boundedness of L_k}
	{\rm \textit{(Boundedness of $L_k$)} }
	Let $f(\cdot)$ be defined in \eqref{problem1} and
	$\{L_k\}_{k=0}^{\infty}$ be the parameters generated in Algorithm \ref{algorithm: Self-concordant Cubic-regularized Newton Method}. 
	Suppose Assumption \ref{assumption:  sub-problem solution} is satisfied. If we take $\beta\in (1,2)$ in Algorithm \ref{algorithm: Self-concordant Cubic-regularized Newton Method}, 
	then 
	\begin{equation}\label{bound on Lk}
	L_k \le \max\{48N, L_0\} \quad \forall k\ge 1.
	\end{equation}
\end{lemma}

Lemma \ref{lemma: Boundedness of L_k} shows that $L_k$ is bounded by a constant that depends on $N$, but not the data $B$. This is a consequence of the self-concordance property of $f(\cdot)$. 
Based on our numerical experience however, it appears that the worst-case 
bound in~\eqref{bound on Lk} is conservative.

Before stating the global convergence result, we introduce some new notations. Let us define the sub-level set
$$
X^0 := \lt\{  w\in \cC  : f(w) \le f(w^0)  \rt\},
$$
where $w^0$ is an initialization for Algorithm \ref{algorithm: Self-concordant Cubic-regularized Newton Method}. Let $w^*$ be an optimal solution of~\eqref{problem1} and 
define
\begin{equation} \label{def: C}
    \begin{aligned}
    C_1 := \sup_{w\in X^0}   
\frac{1}{N} \sum_{i=1}^N \max \Big\{\Big|\frac{\la B_i, w-w^* \ra }{ \la B_i, w\ra} \Big|^3 , \Big|\frac{\la B_i, w-w^* \ra }{ \la B_i, w^*\ra} \Big|^3 
\Big\},& ~~~~~ \text{and}\\
C:= (1/3) \max \Big\{
\Big(1+ \sup_{k\ge 0} L_k/2 \Big) C_1 , ~ f(w^0) - f^*
\Big\}.& 
\end{aligned}
\end{equation}
Note that $C_1$ defined above, is finite\footnote{To see this, note that if $C_{1}=\infty$, then we can find a sequence $\wtd w^k \in X^0 $ such that $f(\wtd w^k)\rightarrow \infty$ as $ k \rightarrow \infty$, 
which is a contradiction to the fact $\wtd w^k \in X^0 $.}. 
In addition, by Lemma~\ref{lemma: Boundedness of L_k}, $C$ is also finite. Let $f^*$ be the optimal objective value of \eqref{problem1}. 
With these notations, we can state the following global convergence result:
\begin{theorem}\label{theorem: global convergence}
	{\rm \textit{(Global sublinear convergence)}}
	Let $\{w^k\}_{k=1}^{\infty}$ be the sequence from Algorithm \ref{algorithm: Self-concordant Cubic-regularized Newton Method}. Let $\dt_k$ be a measure of how close $y^k$ is to minimizing Problem~\eqref{update-yk}, that is, 
	\begin{equation}
		\dt_k := h_f(y^k,w^k) - \min_{y\in \cC}  h_f(y,w^k).
	\end{equation}
	For $i\ge 1$
let $E_i:= \ga_i + \rho_i+\dt_i$. 
	Then for all $k\ge 2$ it holds
	\begin{equation}
	f(w^{k+1}) - f^* \le \frac{30 C}{(k+2)^2} + \frac{1}{(k+1)^3} \sum_{i=1}^k (i+3)^3 E_i \ . \nonumber
	\end{equation}
\end{theorem}

The following corollary states that under mild assumptions on the sequence $\{(\delta_{k}, \gamma_{k},$ $ \rho_{k})\}_{k\geq 1}$, Algorithm~\ref{algorithm: Self-concordant Cubic-regularized Newton Method} converges to the minimum of~\eqref{problem1} at a sublinear rate.
\begin{corollary}\label{cor:inexact-solns}
Under the statements of Theorem \ref{theorem: global convergence}, if the sequences $\{\ga_k\}_{k=1}^\infty$, $\{\rho_k\}_{k=1}^\infty$ are set such that $ \sum_{k=1}^\infty  (\ga_k+\rho_k ) k^3 <\infty $, and the errors $\dt_k$ of solving  sub-problems also satisfy $ \sum_{k=1}^\infty \dt_k k^3 <\infty  $, then Algorithm \ref{algorithm: Self-concordant Cubic-regularized Newton Method} has a global convergence rate $O(1/k^2)$.
\end{corollary}
Particular choices of the parameters $\delta_{k}, \gamma_{k},$ and $ \rho_{k}$ are discussed in Section~\ref{section: Experiments}.

\subsubsection{A local quadratic convergence rate}\label{local-quadratic-CN}

Theorem~\ref{theorem: local convergence} below presents the local convergence properties of Algorithm~\ref{algorithm: Self-concordant Cubic-regularized Newton Method} when $E_{k}=0$ for all $k$.
The proof is based on the analysis of the cubic regularized Newton method in~\cite{nesterov2006cubic} combined with basic properties of self-concordant functions. For this result we assume the matrix $B$ has full row rank -- this implies the Hessian $\na^2 f(w)$ is positive definite for all $w \in \cC$ with $f(w)<\infty$. Our local convergence rate is stated using the norm  $\|\cdot\|_{\na^2 f(w^*)}$, similar to the local convergence rates in~\cite{tran2015composite} for the damped Newton method.

\begin{theorem}\label{theorem: local convergence}
	{\rm \textit{(Local quadratic convergence)}}
	Consider a special case of Algorithm~\ref{algorithm: Self-concordant Cubic-regularized Newton Method} with $\delta_{k}=\ga_k=\rho_k=0$ and let $w^{k+1} = y^k$ in Step 3.
	Let $\{w^k\}_{k=0}^{\infty}$ be the iterates from Algorithm~\ref{algorithm: Self-concordant Cubic-regularized Newton Method}, and let $\bar L:= \sup_{k\ge 0} \{L_k\}$ and $H_* := \na^2 f(w^*)$. 
If for some $K>0$, the following holds:
 $$(6\sqrt{N}+2\bar L) \| w^K - w^* \|_{H_*} \le 1.$$	Then for all $k \ge K$, we have
	\begin{equation}
	\| w^{k+1} - w^* \|_{H_*} \le  (6\sqrt{N} + 2\bar L ) \| w^{k} - w^* \|_{H_*}^2 \ . \nonumber
	\end{equation}
\end{theorem}

\section{Solving the cubic regularized Newton step with Frank-Wolfe}\label{section: using FW to solve  sub-problem}
To implement Algorithm \ref{algorithm: Self-concordant Cubic-regularized Newton Method}, we need an efficient method for (approximately) solving the cubic regularized Newton step~\eqref{update-yk}.
We use the Frank-Wolfe method with away steps~\cite{guelat1986some,lacoste2015global} for this purpose. 
Frank-Wolfe method \cite{frank1956an,Jaggi2013}
is a classical algorithm for constrained convex optimization with a compact constraint set. It is called a ``projection-free" method since at every iteration, we need access to a linear programming (LP) oracle (over the constraint set). The Away-step Frank-Wolfe (AFW) is a variant of the classical Frank-Wolfe method which usually leads to improved sparsity properties and enjoys a global linear convergence rate when the objective function is strongly convex and the constraint set is a polytope\footnote{We note that there are several variants of Frank-Wolfe~\cite[see for example]{freund2017extended} that result in improvements over the basic version of the Frank-Wolfe method. We use the away-step variant for convenience, though other choices are also possible.}. 
Owing to space constraints, a formal description of our AFW algorithm is presented in Section~\ref{sec:append:AFW}. 
For improved practical performance, AFW requires efficient solutions to the LP oracles and efficient methods to store a subset of vertices of the polyhedral constraint set, as we discuss next.

\subsection{Solving the LP oracle}\label{sec:solving-LP-oracle}
We discuss how to efficiently solve the LP oracle
\begin{equation}\label{LP:oracle}
    \min_{w \in {\mathcal C} }~\langle g, w \rangle 
\end{equation} for the different choices of the polyhedral set ${\mathcal C}$ that we consider here. In particular, we provide a description of the set of vertices $V(\cC)$, which leads to a solution of~\eqref{LP:oracle}. 
This can also be used to compute the away step, which requires solving a linear oracle over a subset of vertices.

Proofs of all technical results in this section can be found in Section~\ref{sec:append:LP-oracles}.

\subsubsection{No shape restriction}
The constraint $\cC = \Dt_M$, the $M$-dimensional unit simplex, has $M$ vertices, corresponding to the standard orthogonal bases elements $\{e_i\}_{i=1}^M$. Here,  solving the LP~\eqref{LP:oracle} is equivalent to finding the minimum value of $g$.

\subsubsection{Monotonicity constraint}
For monotone density estimation via formulation~\eqref{problem1} with Bernstein polynomial bases, we consider the following choices for the polyhedral set $\cC$: 
\begin{equation}\label{C: monotone}
\begin{aligned}
\cC^{-} &=& \lt\{    w \in \R^M :~ 1_M^\T w = 1, ~w \ge 0, ~   w_1\ge \cdots \ge w_M \rt\} , \\
\cC^{+} &=& \lt\{    w \in \R^M :~ 1_M^\T w = 1, ~w \ge 0, ~   w_1\le \cdots \le w_M \rt\}   
\end{aligned}
\end{equation}
which correspond to mononotically decreasing and increasing sequences (densities), respectively.
We show below that after a linear transformation, the constraint sets in~\eqref{C: monotone} can be transformed to the standard simplex, for which the LP oracle can be solved in closed form.  
\begin{proposition}\label{proof-proposition-4.1}
(Decreasing)
	Let $U\in \R^{M\times M}$ be an upper triangular matrix with $U_{i,j} = 1/j$ for $j\ge i$ and $i \in [M]$. Then, we can equivalently write $\cC^{-}$ as
	$$ \cC^{-} = U(\Dt_M) := \lt\{    Uy \in \R^M :~ 1_M^\T y = 1, ~y \ge 0 \rt\} .$$
\end{proposition}

For any given $g\in \R^M$, the LP~\eqref{LP:oracle} (with $\cC = \cC^{-}$),
is equivalent to solving
\begin{equation}
\min_y ~ \la U^\T g, y \ra  ~~~ {\rm s.t.} ~~~ 1_M^\T y = 1, ~ y \ge 0 \nonumber
\end{equation}
which is equivalent to finding the minimum value in the vector $U^\T g$. By exploiting the structure of matrix $U$, computing $U^\T g$ only needs $O(M)$ operations. 

Similarly, when $\cC=\cC^{+}$ (increasing sequence), we arrive at the following proposition.
\begin{proposition}
(Increasing)
	Let $U\in \R^{M\times M}$ be a lower triangular matrix with $U_{i,j} = 1/(M-j+1)$ for $i\ge j$ and $j \in [M].$ Then we have
	$$ \cC^{+} = U(\Dt_M) = \lt\{    Uy \in \R^M :~ 1_M^\T y = 1, ~y \ge 0 \rt\} .$$
\end{proposition}
The LP oracle~\eqref{LP:oracle} (for $\cC = \cC^+$) can be solved in a manner similar to the case of $\cC^-$ --- we omit the details for brevity.

\subsubsection{Concavity constraint}\label{sec:concavity-constraint}
For concave density estimation with a Bernstein polynomial bases (see Section~\ref{sec: shape-constraints}), in Problem~\eqref{problem1},
we consider $\cC=\cC^{\wedge}$, where 
\begin{equation}\label{C: concavity}
\cC^{\wedge} = \lt\{    w \in \R^M :~ 1_M^\T w = 1, ~w \ge 0, ~   w \text{ is concave} \rt\}.
\end{equation}
The following proposition provides a description of the extreme points of the set $\cC^{\wedge}$.
\begin{proposition}\label{proposition:4point3}
	(Concave) $\cC^{\wedge}$ has $M$ vertices, which are given by
	\begin{equation}
	\begin{aligned}
	v_1 &=  \frac{2}{M(M-1)} (0,1,...,M-1)^\T , ~~ v_M = \frac{2}{M(M-1)} ( M-1,M-2,...,1,0)^\T , \\
	v_j&=(0, p_j, ...,(j-1)p_j, (M-j-1)q_j,  (M-j-2)q_j, ..., q_j,0)^\T ~~ 2 \le j \le M-1, 
	\end{aligned}
	\nonumber
	\end{equation}
	where $p_j = 2/((M-1)(j-1))$ and $q_j = 2/((M-1)(M-j))$ for all $2 \le j \le M-1$.
\end{proposition}

Let $V$ be the matrix whose columns are $v_1,\ldots,v_M$. 
For any given $g\in \R^M$, the LP oracle
$\min_{w\in \cC^{\wedge}}  \la g,w \ra $ 
amounts to solving the following problem 
\begin{equation}
\min_y ~ \la V^{\T} g, y \ra  ~~~ {\rm s.t.} ~~~ 1_M^\T y = 1, ~ y \ge 0 \nonumber
\end{equation}
which is equivalent to finding the minimum entry in the vector $V^{\T} g$. A careful analysis shows that computing $V^{\T} g$ only requires $O(M)$ operations.

\subsubsection{Convexity constraint}
Similar to Section~\ref{sec:concavity-constraint} (concavity constraint), 
we consider the case when $\cC$ in~\eqref{problem1} corresponds to a convex sequence, resulting in a convex density estimate. Here, in Problem~\eqref{problem1}, we take $\cC=\cC^{\vee}$ where, 
\begin{equation}\label{C: convexity}
\cC^{\vee} = \lt\{    w \in \R^M :~ 1_M^\T w = 1, ~w \ge 0, ~   w \text{ is convex} \rt\}. 
\end{equation} 
For any $k\ge 0$, let $0_k$ denote a vector in $\R^k$ with all coordinates being $0$. The following proposition describes the vertices of $\cC^{{\vee}}$. 
\begin{proposition}\label{proposition: vectices concavity}
(Convex)	$\cC^{\vee} $ has $2M$ vertices, which are given by
	\begin{equation}
	\cup_{k=1}^M \lt\{    a_k^{-1}  (0_{M-k}^\T, 1, 2, \dots , k)    ~,~   a_k^{-1} (k, k-1,  \dots , 2,1, 0_{M-k}^\T)   \rt\} \nonumber
	\end{equation}
	where $a_k = k(k+1)/2$ for all $k\in [M]$.
\end{proposition}

By Proposition~\ref{proposition: vectices concavity}, solving the LP oracle~\eqref{LP:oracle} with $\cC = \cC^{\vee}$ amounts to computing
\begin{equation}
\min \Big\{     \min_{1\le k \le M} \Big\{ a_k^{-1} \sum_{i=1}^k i g_{i+M-k} \Big\}, ~   
\min_{1\le k \le M} \Big\{ a_k^{-1} \sum_{i=1}^k i g_{k+1-i} \Big\}
   \Big\}  \ . \nonumber
\end{equation}
With a careful cost accounting, the above can be computed within $O(M)$ operations.

\subsubsection{Concavity and monotonicity restriction}
We consider the cases when $\cC$ corresponds to a combination of concavity and  monotonity constraints:
\begin{eqnarray}\label{C: concave and monotone}
\cC^{\wedge +} = \lt\{    w \in \R^M :~ 1_M^\T w = 1, ~w \ge 0, ~   w \text{ is increasing and concave} \rt\}  ,  \nonumber\\
\cC^{\wedge  -} = \lt\{    w \in \R^M :~ 1_M^\T w = 1, ~w \ge 0, ~   w \text{ is decreasing and concave} \rt\}  .  \nonumber
\end{eqnarray}
Following our discussion in Section~\ref{sec: shape-constraints}, if we set $\cC=\cC^{\wedge +}$ in~\eqref{problem1} and use Bernstein polynomial bases, we obtain a maximum likelihood density estimate that is concave and increasing. Similarly, $\cC^{\wedge -}$ leads to a density that is concave and decreasing. 
Propositions~\ref{proposition: concave increasing restriction} and~\ref{proposition: concave decreasing restriction} below describe the set of vertices of $\cC^{\wedge +}$ and $\cC^{\wedge -}$, respectively.

	\begin{proposition}\label{proposition: concave increasing restriction} (Concave and increasing)
		$\cC^{\wedge  +}$ has $M$ vertices, which are given by 
		\begin{equation}
		\begin{aligned}
		v_1 &=M^{-1} 1_M, ~~ v_M = r_M (0,1,2,...,M-1)^\T , \\
		v_i &= r_i \lt(0,1,...,i-2,   i-1, (i-1)1_{M-i}^\T \rt)^\T, \quad 2\le i \le M-1 , 
		\end{aligned}
		\nonumber
		\end{equation}
		where $r_i = 2/((2M-i)(i-1))$ for $2\le i \le M$.

	\end{proposition}

Similar to Proposition~\ref{proposition: concave increasing restriction}, we present the corresponding proposition for $\cC^{\wedge -}$, whose proof is omitted due to brevity. 
	\begin{proposition}\label{proposition: concave decreasing restriction} (Concave and decreasing)
	$\cC^{\wedge -}$ has $M$ vertices, which are given by 
	\begin{equation}
	\begin{aligned}
	v_1 &=M^{-1} 1_M, ~~ v_M = r_M (M-1, M-2, ...., 1,0)^\T , \\
	v_i &= r_i \lt( (i-1)1_{M-i}^\T, i-1, i-2, ..., 1,0 \rt)^\T \quad 2\le i \le M-1 , 
	\end{aligned}
	\nonumber
	\end{equation}
	where $r_i = 2/((2M-i)(i-1))$ for $2\le i \le M$. 
\end{proposition}
In light of Propositions~\ref{proposition: concave increasing restriction} and \ref{proposition: concave decreasing restriction}, we can efficiently compute the corresponding LP oracles~\eqref{LP:oracle} with $\cC=\cC^{\wedge +}$ and $\cC=\cC^{\wedge -}$.

\subsubsection{Convexity and monotonicity restriction}
We consider the cases when $\cC$ corresponds to a combination of convexity and  monotonity constraints. Here, the constraint set $\cC$ is one of the following:
\begin{eqnarray}\label{C: convex and monotone}
\cC^{\vee+} = \lt\{    w \in \R^M :~ 1_M^\T w = 1, ~w \ge 0, ~   w \text{ is convex and increasing}\rt\}  ,  \nonumber\\
\cC^{\vee -} = \lt\{    w \in \R^M :~ 1_M^\T w = 1, ~w \ge 0, ~   w \text{ is convex and decreasing} \rt\}  .  \nonumber
\end{eqnarray}
The following proposition describes the set of vertices of $\cC^{\vee +}$. 

	\begin{proposition}\label{proposition: convex increasing restriction} (Convex and increasing)
	$\cC^{\vee +}$ has $M$ vertices, which are given by 
	\begin{eqnarray}
	v_i = \frac{2}{i(i+1)} \lt( 0^\T_{M-i}, 1,2,..., i \rt)^\T,~1\le i \le M-1;~~\text{and}~~ v_M = \frac{1}{M} 1_M \ .  \nonumber
	\end{eqnarray}
\end{proposition}

Similarly, we have the following proposition for $\cC^{\vee -}$, whose proof is omitted. 
	\begin{proposition}\label{proposition: convex decreasing restriction} (Convex and decreasing)
	$\cC^{\vee -}$ has $M$ vertices, which are given by 
	\begin{eqnarray}
	v_i = 
		\frac{2}{i(i+1)} \lt( i, i-1, ..., 2,1,0^\T_{M-i} \rt)^\T, ~~ 1\le i \le M-1;~~\text{and}~~
	  v_M = \frac{1}{M} 1_M. \nonumber
	\end{eqnarray}
\end{proposition}

Using Propositions~\ref{proposition: convex increasing restriction}  and~\ref{proposition: convex decreasing restriction}, one can maintain the active vertices and solve the LP oracles, in a manner similar to our earlier discussions.

\subsubsection{Unimodality constraint}\label{sec:unimodal-update}
We consider the problem of learning a mixture of Bernstein polynomial bases elements, such that the resulting density $x \mapsto g_{B}(x)$ is unimodal. Recall that the density $x \mapsto g_{B}(x)$ supported on $[0,1]$ is unimodal if it is increasing on $[0,a]$ and then decreasing on $[a,1]$ with an unknown mode location at $a \in (0,1)$. 
If the sequence $m \mapsto w_m$ is unimodal\footnote{That is, $m \mapsto w_{m}$ is increasing at the beginning and then decreasing, with a maximum located somewhere in $\{1, \ldots, M\}$.}, then $x \mapsto g_{B}(x)$ will also be unimodal~\cite{carnicer1993shape,turnbull2014unimodal}. If we know the location of the mode in $m \mapsto w_{m}$, then the unimodal shape restriction can be represented via a convex polyhedral constraint on $w$. A unimodal shape restriction on $w$ where the modal location is unknown, can be represented by a union of polyhedral sets and is not a convex constraint. 

Consider sequences $\{w_{i}\}_{1}^{M}$ that have a mode (or maximum) at a pre-specified location $k \in [1,M]$. In this case, $\cC$ is given by the following convex polyhedron:
\begin{equation}\label{C: unimodal}
\cC^{um}_k = \{ w\in \R^M ~:~ 1^\T_M w = 1, ~ w\ge 0, ~w_{1} \le \cdots \le w_k , ~ w_k \ge \cdots \ge w_{M}\}. 
\end{equation}
To estimate $g_{B}(x)$ under a unimodality constraint on $w$ where the mode location is unknown, we solve Problem~\eqref{problem1} with $\cC = \cC_{k}^{um}$ for each $k \in [M]$. Then we choose the value of $k$ which minimizes the negative log-likelihood or equivalently, the objective $f(w)$ in~\eqref{problem1}. This leads to the MLE for the mixture density $g_{B}(x)$ under the assumption of unimodality. 

We now discuss how to solve the LP oracle when $w \in \cC^{um}_k$ for a fixed $k$. The following result lists all vertices of $\cC^{{um}}_k$. 
\begin{proposition}\label{proposition-unimodality-formula}
	$\cC^{um}_k$ has $k(M-k+1)$ vertices, which are given by: 
$v^{k_1,k_2}$ for $1\le k_1 \le k \le k_2\le M$, 
	where the $i$-th element in $v^{k_1,k_2} \in \R^M$ is: 
	$$v^{k_1,k_2}_i = \begin{cases}0 & \text{for $i<k_1$ or $i>k_2$} \\
	1/(k_2-k_1+1) & \text{for $k_1 \le i \le k_2$}. 
	\end{cases} $$
\end{proposition}

Let $U$ be the matrix with columns being the vertices of $\cC^{um}_k$. For a given $g\in \R^M$, solving the linear oracle $\min_{w\in \cC_k^{um}} \la g , w \ra$ amounts to finding the minimum entry in the vector $U^\T g$. This requires $O(k(M-k))$ operations with a careful manipulation.

\section{Approximation guarantees for the Kiefer-Wolfowitz MLE}\label{sec: approximation-guarantees}
We  consider a special case of~\eqref{crit-our} when the $i$-th mixture component is a Gaussian density with unit variance and mean $\mu_{i}$  for $i \in [M]$. In this case, Problem~\eqref{crit-our} can be interpreted as a finite dimensional version of the original Keifer-Wolfowitz nonparametric MLE problem~\cite{kiefer1956consistency,koenker2014convex}:
\begin{equation}\label{infinite-dim-primal}
p^*:=\min_F ~~~-\sum_{i=1}^N \log \Big(    \int \varphi (X_i - \mu) d {\mathcal Q}(\mu)  \Big)
\end{equation}
where ${\mathcal Q}$ runs over all mixing distributions (of the location parameter, $\mu$) on $\R$.
Note that a solution to~\eqref{infinite-dim-primal}
admits a finite dimensional solution~\cite{koenker2014convex}, which  
can be described via a dual formulation of~\eqref{infinite-dim-primal}, as presented below.
\begin{theorem}\label{theorem: infinite duality}
[Theorem 2 of \cite{koenker2014convex}]
{\rm \textit{(Duality of~\eqref{infinite-dim-primal})}}
Problem~\eqref{infinite-dim-primal} has a minimizer ${\mathcal Q}^*$, which is an atomic probability measure with at most $N$ atoms. The locations, $\{\mu_j^*\}$, and the masses, $\{w_j^*\}$, at these $N$ locations can be found via the following dual characterization:
\begin{equation}\label{infinite-constraints-dual}
\max_{\nu\in \R^N_{++}}~   \sum_{i=1}^N  \log \nu_i ~~~\text{s.t.} ~~~
\sum_{i=1}^N \nu_i \varphi (X_i - \mu ) \le N ~~ \forall \mu\in \R.  
\end{equation}
The solution $\nu^*$ of~\eqref{infinite-constraints-dual} 
satisfies the extremal equations $\sum_{j} \varphi (X_i - \mu_j^* ) w_j^* = 1/{\nu_i^*}$, $i\in [N]$, 
and $\mu^*_j$'s are exactly those values of $\mu$ at which the dual constraint is active.
\end{theorem}
The dual objective in~\eqref{infinite-constraints-dual} contains $N$ decision variables, but the constraints are infinite dimensional. While an optimal solution to~\eqref{infinite-dim-primal} contains $N$ atoms, these locations are not known; and it is not clear how to efficiently identify these locations without solving the infinite-dimensional problem. 
\cite{lindsay1983geometry,groeneboom2008support} present algorithms to approximate this infinite dimensional problem~\eqref{infinite-dim-primal}. In particular, \cite{lindsay1983geometry} proposes a vertex-directional method and provides bounds for the computed solution. 
We explore some properties of the commonly-used discretized version   of~\eqref{infinite-dim-primal} where ${\mathcal Q}$ is supported on a pre-specified discrete set of atoms between
$\min_{i} X_i$ and $\max_{i} X_i$---this version can be written in the form \eqref{crit-our} and solved by Algorithm~\ref{algorithm: Self-concordant Cubic-regularized Newton Method}. Specifically, we present bounds on how well a solution to the discretized problem approximates the infinite-dimensional counterpart in Propositions~\ref{prop: dicrete-dual} (in terms of optimal primal objective values) and~\ref{prop: error bound dual} (in terms of dual optimal solutions).

Note that if $\cG = \{\hat \mu_1,...,\hat \mu_M \} $ ($\hat \mu_1 < \cdots < \hat \mu_M$) denotes a pre-specified set of $M$ atoms, then the discrete version of \eqref{infinite-dim-primal} is 
\begin{equation}\label{discrete-primal}
\hat{p}_{\mathcal G} := \min_{w\in \R^M} ~ - \sum_{i=1}^N \log\Big(
\sum_{j=1}^M w_j \varphi (X_i - \hat \mu_j )
\Big) ~~~\text{s.t.} ~~~w \ge 0, ~ \sum_{i \in [M]} w_{i} = 1. 
\end{equation}
The following proposition presents a dual of~\eqref{discrete-primal}.
\begin{proposition}\label{prop: dicrete-dual}
		{\rm \textit{(Duality for finite dimensional approximation)}}
	The optimal objective of \eqref{discrete-primal} equals the optimal objective of the following dual problem:
	\begin{equation}\label{discrete-dual}
	\max_{\nu \in \R^N} ~     \sum_{i=1}^N \log \nu_i ~~~\text{s.t.}~~ \sum_{i=1}^N \nu_i \varphi (X_i- \hat \mu_j) \le N ~~ \forall j\in [M], ~ \nu \ge 0     \ . 
	\end{equation}
 Furthermore, an optimal solution $\hat w$ of \eqref{discrete-primal} and the optimal solution $\hat \nu$ of \eqref{discrete-dual} are linked by $ \sum_{j=1}^M \hat w_j \varphi (X_i - \hat \mu_j )= {1}/{\hat \nu_i}$ for all $i\in [N]$. 
	If for some $j\in [M]$ we have a strict inequality $\sum_{i=1}^N \nu_i\varphi(X_i - \hat \mu_j)<N$, then
	by complementary slackness, we have $\hat w_j = 0$. 
\end{proposition}

The optimal objective value $\hat{p}_{\mathcal G}$ of \eqref{discrete-primal}  is larger than the optimal objective value $p^*$ of \eqref{infinite-dim-primal}. 
In the following, we present an upper bound on $\hat{p}_{\mathcal G} - p^*$, which depends upon the discrete set of atoms ${\mathcal G}$.
To derive our result, we make use of the observation that $\varphi$ is Lipschitz continuous over $\R$, that is, for all $t_1,t_2 \in \R$, $| \varphi (t_1) -\varphi (t_2) | \le 
(1/{\sqrt{2\pi e}}) |t_1 - t_2|$. 
\begin{proposition}\label{prop: optimality gap}
	{\rm \textit{(Primal optimality gap)}}
	Let
	 $\cG = \{\hat \mu_1,...,\hat \mu_M\}$  (with $\hat \mu_1 < \cdots < \hat \mu_M$) be a pre-specified set of discretized locations satisfying $\hat \mu_1 \le \min_{i\in [N]} X_i \le \max_{i\in [N]} X_i \le \hat \mu_M $.
	  Let $\hat p_\cG$ be the optimal value of \eqref{discrete-primal}, and let $p^*$ denote the optimal value of the infinite dimensional problem \eqref{infinite-dim-primal}. Let $\hat \nu$ be the optimal solution of \eqref{discrete-dual}.
	Suppose $\Gamma>0$ satisfies 
	$\Gamma\ge \max_{\mu \in \R} \sum_{i=1}^N \hat \nu_i \varphi (X_i - \mu)   $.
	Then it holds
	\begin{equation}\label{bound1}
	p^* ~\le ~\hat p_\cG ~\le~ p^* + N\log(\Gamma/N).
	\end{equation}
Let	$\Dt \cG$ be the largest spacing among ordered atoms in $\cG$, that is, $\Dt \cG := \max_{i} (\hat \mu_{i+1} - \hat \mu_i)$.
Then we have
	\begin{equation}\label{bound2}
	p^* ~\le ~\hat p_\cG ~\le~ p^* + N \log \Big( 1+ \frac{\Dt \cG}{\sqrt{8\pi e}} \frac{1}{N} \sum_{i=1}^N \hat \nu_i  \Big) \ .
	\end{equation}
\end{proposition}

As a consequence of Proposition~\ref{prop: optimality gap}, if $\nabla {\mathcal G} \approx 0$, then $p^* \approx \hat{p}_{\mathcal G}$.
Using Proposition~\ref{prop: optimality gap}, we present the following bound on the proximity of the dual solutions corresponding to the finite and infinite-dimensional problems.

\begin{proposition}\label{prop: error bound dual}
		{\rm \textit{(Error bound on dual optimal solutions)}}
		Let $\cG$, $\hat p_\cG$ and $p^*$ be the same as in Proposition \ref{prop: optimality gap}. 
		Let $\hat \nu$ and $\nu^*$ be the optimal solutions of \eqref{discrete-dual} and \eqref{infinite-constraints-dual} respectively. Then it holds
		\begin{equation}\label{dual value bound}
		\Big(  \sum_{i=1}^N {(\nu_i^* - \hat \nu_i)^2}/{\hat \nu_i^2}  \Big)^{1/2} \le \rho^{-1}(\hat p_{\cG} - p^*), 
		\end{equation}
		where, $\rho^{-1}(t)$ denotes the inverse of the (strictly increasing) function $\rho(t)= t - \log(1+t)$, defined on $[0, \infty)$. 
\end{proposition}

\section{Experiments}\label{section: Experiments}

In this section\footnote{Some implementation details of Algorithm~\ref{algorithm: Self-concordant Cubic-regularized Newton Method} can be found in Section~\ref{sec:append:experimental-details}.}, we compare the numerical performance of our proposed Algorithm~\ref{algorithm: Self-concordant Cubic-regularized Newton Method} versus other benchmarks for Problem~\eqref{problem1} for different choices of $\cC$ on both synthetic and real datasets. 
Our algorithms are written in Python 3.7.4 and can be found at

{{\small\url{https://github.com/wanghaoyue123/Nonparametric-density-estimation-with-shape-constraints}}}\\
 All computations 
are performed on MIT Sloan's engaging cluster with 1 CPU and 32GB RAM.

\subsection{Synthetic data experiments}
\subsubsection{Solving Problem~\eqref{crit-our}}\label{sec:compute-unconstrained}
We first consider the problem of estimating mixture proportions with 
no shape constraint; and compare our algorithm with several benchmarks on synthetic datasets in terms of optimizing~\eqref{crit-our}.\\
\noindent {\bf Setup.} We generate iid samples $X_1,\ldots,X_N$ from a Gaussian mixture density with $5$ components: $g_{B}^\dagger(x)=\sum_{i=1}^5 w_i^\dagger N (\mu_i^\dagger , (\sigma_i^\dagger)^2)$
with mixture proportions $( w_1^\dagger,\ldots, w_5^\dagger)$= $(0.6, 0.05, $ $0.15, 0.1, 0.1) $, location parameters $(\mu_1^\dagger,\ldots,\mu_5^\dagger)= (0, 4,$ $ 5.5, -3.5, -4.5)$, and 
standard deviations $(\sigma_1^\dagger,\ldots,\sigma_5^\dagger) = (1, 0.5, 1, 0.25, 0.25)$. 
To approximate the underlying density $g_{B}^\dagger(x)$, we consider a mixture of
$M$ Gaussian densities: The basis elements are given by 
$\psi_{i}(x)  = \varphi((x - \mu_{i})/\sigma)$
for $i \in [M]$, and $\{\mu_i\}_{1}^M$ form an equi-spaced grid between the smallest and largest values in $X_1,\ldots,X_N$. 
Note that the true underlying density $g_{B}^\dagger(x)$ is not a member of $g_{B}(x)=\sum_{i \in [M]} w_{i} \psi_{i}(x)$, as each 
basis element $\psi_i$ has variance $\sigma^2$.

\noindent {\bf Benchmarks and metrics.} We 
compare Algorithm \ref{algorithm: Self-concordant Cubic-regularized Newton Method} with (i) the commercial solver Mosek~\cite{andersen2000mosek} using CVXPY interface;
(ii) the sequential quadratic programming based method (MIXSQP) introduced in \cite{kim2020fast}, (iii) the Frank-Wolfe method (FW) in \cite{dvurechensky2020self}, and (iv) proximal gradient descent (PGD)\footnote{For FW and PGD we use the public code available at~\url{https://github.com/kamil-safin/SCFW}.} 
in \cite{tran2015composite}. 

In Table \ref{table: compare timing synthetic}, we present the runtimes of these five methods for different values of $M\in \{200,500,1000\}$ and $N\in\{10^5,2\times 10^5, 5\times 10^5\}$. 
We set the standard deviation of the mixture components to $\sigma = 0.2$ (this value is chosen so that the resulting density leads to a good visual fit). 
For each method, Table~\ref{table: compare timing synthetic} presents the following metrics: (a)~\textit{Time}: the time (s) needed for an algorithm to output a solution.
(b)~\textit{Relative Error}: For an output solution $\tilde{w}$ (computed by an algorithm), we define the relative error as 
$$
\textit{Relative Error} = (f(\tilde{w}) - f^* ) / \max \lt\{1, |f^*|\rt\}  
$$ 
where, $f^*$ is the (estimated) optimal objective value as obtained by 
 Mosek\footnote{For the examples where Mosek fails to output a solution  due to insufficient memory (we limit memory to 32GB), we run Mosek with 128GB memory to obtain the value of $f^*$.}. 
 In Table~\ref{table: compare timing synthetic}, the time and relative error for each method is reported as time/error; and for Mosek, we only report the runtime as the relative error is zero.
In our experiment, we use 
Mosek with its default setting.
For MIXSQP we also use the default setting, but set a loose convergence tolerance of $0.03$. The FW and PGD algorithms are terminated as soon as they reach a solution with a relative error smaller than $10^{-4}$. In Algorithm~\ref{algorithm: Self-concordant Cubic-regularized Newton Method}, we solve every cubic Newton step with AFW; 
we terminate our method as soon as its relative error falls below $10^{-4}$. See Section~\ref{sec:append:experimental-details} for additional implementation details of Algorithm~\ref{algorithm: Self-concordant Cubic-regularized Newton Method}.

Table \ref{table: compare timing synthetic} suggests that our proposed approach can be faster than other methods---we obtain approximately a 3X-5X improvement over the next best method (MIXSQP), to reach a solution with relative error $\sim 10^{-4}$--$10^{-5}$. 
We note that~\cite{kim2020fast} use a low-rank approximation of $B$ for improved computational performance. We do not use low-rank approximations in our implementations, but note that they may lead to further improvements. 
The interior point solver of Mosek also works quite well; and we obtain up to a 10X improvement over Mosek on some instances. Compared to Mosek, our method appears to be more memory efficient---see for example, the instance with $M=1000$ and $N=5\times 10^5$.

\begin{table}[h] 
	\centering 
	\resizebox{\textwidth}{!}{ 
		\begin{tabular}{|c|c|ccccc|} 
\hline 
$N$                     & $M$    & Mosek   & MIXSQP~\cite{kim2020fast}       & FW~\cite{dvurechensky2020self}        & PGD~\cite{tran2015composite} & Algorithm~\ref{algorithm: Self-concordant Cubic-regularized Newton Method}          \\ 
  && time (s) & time(s)/error & time(s)/error& time(s)/error& time(s)/error\\
\hline 
\multirow{3}{*}{1e5} & 200  & 63.2 & 23.8/1.1e-04  & 189.9/1.0e-04 & 101.8/9.6e-05 & 8.6/8.4e-05   \\ 
                      & 500  & 102.9 & 76.6/1.3e-04  & 188.5/1.0e-04 & 270.0/9.9e-05 & 23.9/8.7e-05  \\ 
                      & 1000 & 186.7 & 186.4/8.2e-05  & 414.4/9.7e-05 & 575.0/9.8e-05 & 60.7/9.8e-05  \\ \hline 
\multirow{3}{*}{2e5} & 200  & 114.1 & 47.5/9.1e-05  & 477.5/1.0e-04 & 323.9/9.9e-05 & 13.5/5.9e-05  \\ 
                      & 500  & 207.4 & 165.4/9.2e-05  & 1075.5/1.0e-04 & 856.4/9.8e-05 & 36.7/6.5e-05  \\ 
                      & 1000 & 372.2 & 329.2/9.2e-05  & 1851.9/9.9e-05 & 1729.0/9.9e-05 & 98.7/6.4e-05  \\ \hline 
\multirow{3}{*}{5e5} & 200  & 361.5 & 110.2/1.2e-04  & 1183.6/1.0e-04 & 1094.8/1.0e-04 & 27.1/5.6e-05  \\ 
                      & 500  & 608.6 & 329.1/1.2e-04  & 2609.5/1.0e-04 & 3186.1/1.0e-04 & 88.9/5.5e-05  \\ 
                      & 1000 & O.M. & 738.0/1.2e-04  & 2781.9/9.8e-05 & 6747.7/1.0e-04 & 236.5/5.9e-05 \\ \hline 
\end{tabular}
	}
	\caption{\small Comparison of timings for density estimation with no shape constraints i.e., Problem~\eqref{crit-our}.
		For each method, we present time(s) and (relative) error as defined in the text. Here ``O.M." refers to the instance where Mosek runs out of memory (32G).}
		\label{table: compare timing synthetic}
\end{table}

\subsubsection{Solving Problem~\eqref{problem1}: with shape constraints}\label{sec:compute-constrained}
For density estimation with shape constraints, we present results for two settings: one with a concave density, and the other involving a density that is convex and increasing. 

We generate $N$ iid samples from the true density $g^\dagger_{B}(x) = \sum_{m=1}^5 w_m^\dagger \td{b}^\dagger_m (x)$, which is a finite mixture of 5 Beta-densities. The concave density $g^\dagger_{B}(x)$ is obtained by choosing
$(w_1^\dagger, ..., w_5^\dagger) = (0.05, 0.3, 0.3, 0.3, 0.05)$. 
To obtain a density $g^\dagger_{B}(x)$ which is increasing and convex, we use
$(w_1^\dagger, ..., w_5^\dagger)=(0.05, $ $0.05, 0.1, 0.25, 0.55)$. 
For each case, we learn the density $g_{B}^\dagger$ using a mixture of $M$ Bernstein densities with knots located uniformly in the interval $[0,1]$---we obtain the mixture weights by solving Problem~\eqref{problem1} (with the corresponding shape constraint). 
See Table~\ref{table: compare timing synthetic2} for results.

Unlike~\eqref{crit-our}, there are no specialized benchmarks for the shape constrained problem~\eqref{problem1}.
We only compare our method with Mosek.
When Mosek does not run into memory issues, our method outperforms Mosek in runtime by a factor of 3X$\sim$20X, with a reasonably good accuracy $10^{-4}\sim 10^{-5}$. For larger problems, our approach works well but Mosek would not run due to memory issues.

\begin{table}[h!]
\centering
\scalebox{1}{\begin{tabular}{cc}
		\begin{tabular}{|c|c|cc|} 
			\multicolumn{4}{c}{Concave density} \\ 
			\hline 
			$N$                     & $M$    & Mosek   & Algorithm~\ref{algorithm: Self-concordant Cubic-regularized Newton Method}      \\ 
			   && time (s) & time(s)/error\\ \hline 
			\multirow{3}{*}{1e5} & 200  & 60.9 & 3.5/1.1e-05   \\ 
			& 500  & 107.4 & 11.9/9.6e-06  \\ 
			& 1000 & 225.8 & 34.8/6.2e-05  \\ \hline 
			\multirow{3}{*}{2e5} & 200  & 68.1 & 3.7/6.2e-05  \\ 
			& 500  & 127.8 & 15.5/3.6e-05  \\ 
			& 1000 & 280.6 & 44.3/6.3e-05  \\ \hline 
			\multirow{3}{*}{5e5} & 200  & 200.4 & 10.0/6.9e-05  \\ 
			& 500  & 440.3 & 44.6/5.2e-05  \\ 
			& 1000 & O.M. & 119.3/9.4e-05 \\ \hline 
		\end{tabular} &
		\begin{tabular}{|c|c|cc|} 
			\multicolumn{4}{c}{Increasing and convex density} \\ 
			\hline 
			$N$                     & $M$    & Mosek   & Algorithm~\ref{algorithm: Self-concordant Cubic-regularized Newton Method}            \\ 
			   && time (s) & time(s)/error\\
			\hline 
			\multirow{3}{*}{1e5} & 200  & 24.9 & 5.3/2.9e-05   \\ 
			& 500  & 66.4 & 19.1/2.7e-05  \\ 
			& 1000 & 140.7 & 58.1/1.3e-05  \\ \hline 
			\multirow{3}{*}{2e5} & 200  & 78.6 & 9.2/3.6e-05  \\ 
			& 500  & 139.7 & 36.3/3.8e-05  \\ 
			& 1000 & 248.6 & 98.8/1.6e-05  \\ \hline 
			\multirow{3}{*}{5e5} & 200  & 255.3 & 22.2/1.2e-05  \\ 
			& 500  & O.M. & 67.5/1.7e-05  \\ 
			& 1000 & O.M. & 218.4/2.0e-05 \\ \hline 
		\end{tabular}
			\end{tabular}}
		\caption{\small Comparison of timings for density estimation with shape constraints--we consider Problem~\eqref{problem1} with two shape constraints: [left] concave, [right] convex and increasing. Mosek runs into memory issues (O.M.) for some larger instances.}
		\label{table: compare timing synthetic2}
\end{table}

\subsection{Real data example}\label{sec:real-data}
We illustrate our approach on real-world data from the US Census Bureau 
to demonstrate examples of shape-constrained density estimation for some covariates. The data is publicly available in US Census Planning Database 2020 that provides a range of demographic, socioeconomic, housing and census operational data~\cite{Bureau2020}. It includes covariates from Census 2010 and 2014-2018 American Community Survey (ACS), aggregated at both Census tract level as well as block level in the country. We consider tract level data, which has approximately $74,000$ samples and $500$ covariates. We consider the marginal density estimation for 4 different features (denoted as feature-1,..,feature-4)\footnote{The features correspond to: feature-1: number of children under age 18 living in households below poverty line. feature-2: percentage of population who are citizens of the United States at birth. feature-3: percentage of population that identify as ``Mexican'', ``Puerto Rican'', ``Cuban'', or ``another Hispanic, Latino, or Spanish origin''. feature-4: Percentage of children of age 3 and 4 that are enrolled in school. All these numbers pertain to the ACS database.  All features have been normalized within the interval $[0, 1]$ before performing density estimation.} to illustrate our density estimation procedure with different shape constraints, such as decreasing, increasing, convex and concave---all of which can be handled by Problem~\eqref{problem1}.

\begin{figure}[b!]
        \centering
\scalebox{0.9}{\begin{tabular}{cc}
decreasing density & increasing density \\
            \includegraphics[width=0.5\textwidth]{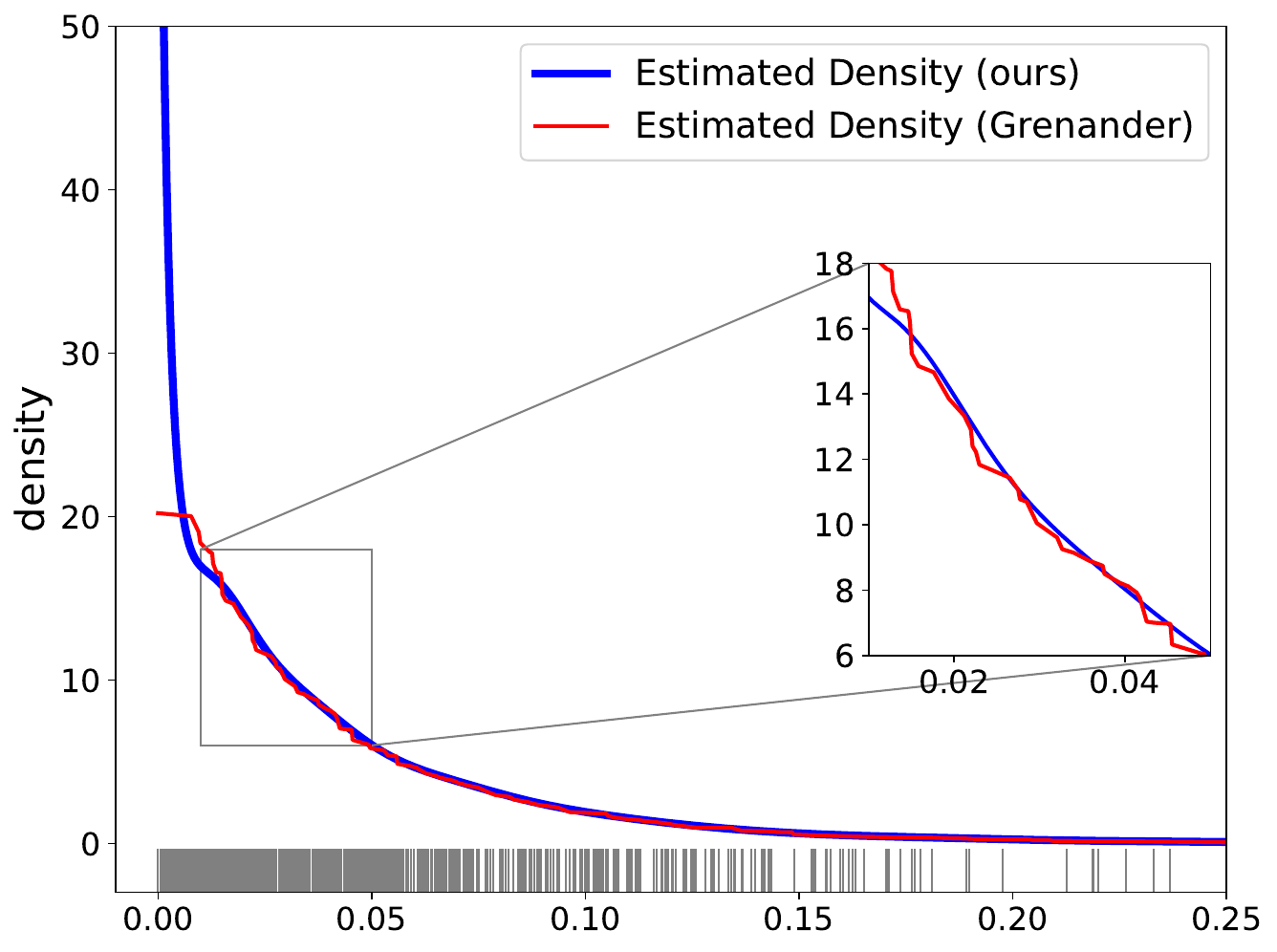}&
        \includegraphics[width=0.5\textwidth]{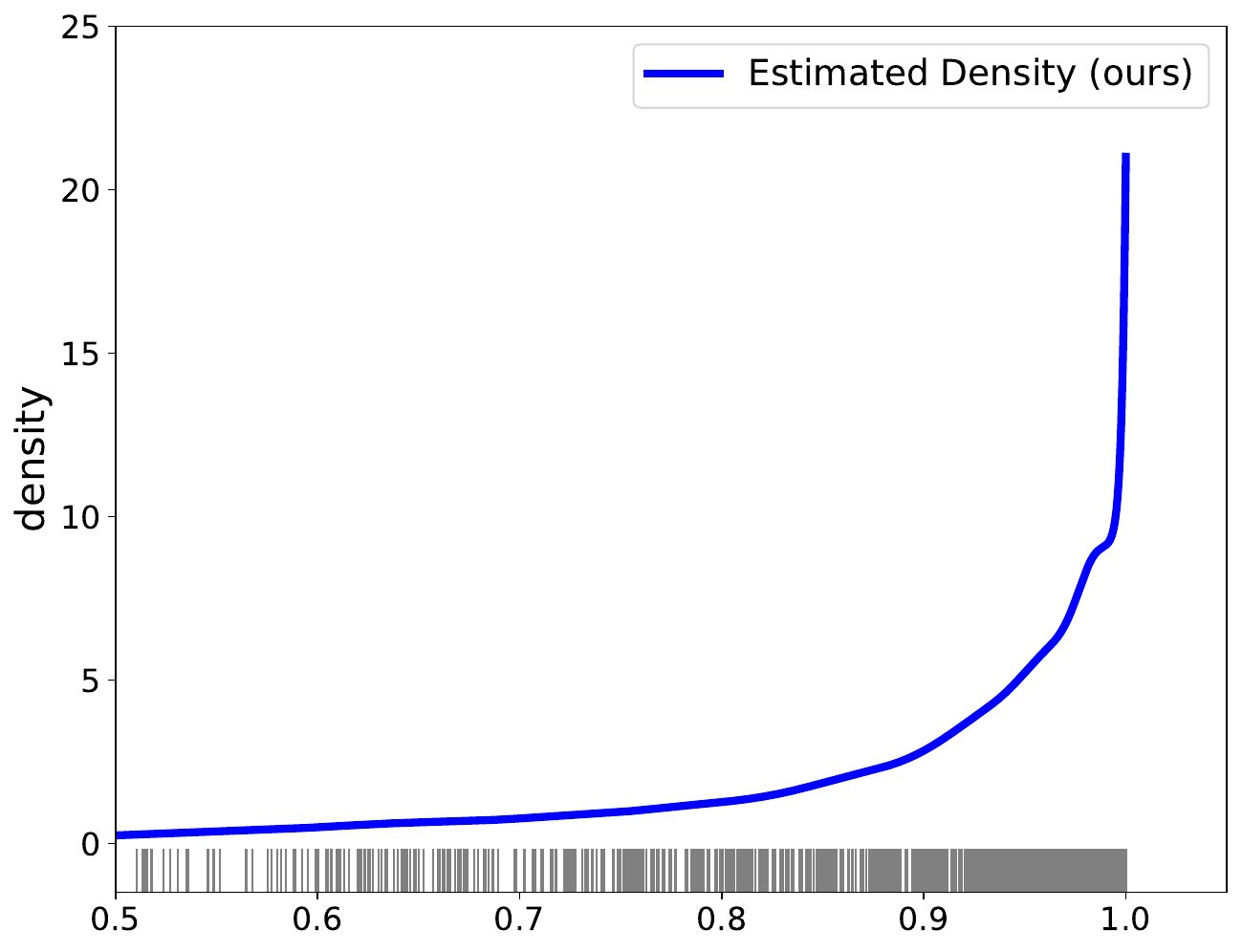}\\
        (a) feature-1 & (b) feature-2  \vspace{1em} \\
    ~convex density/decreasing density  & ~no shape constraint/concave density \\
\includegraphics[width=0.5\textwidth]{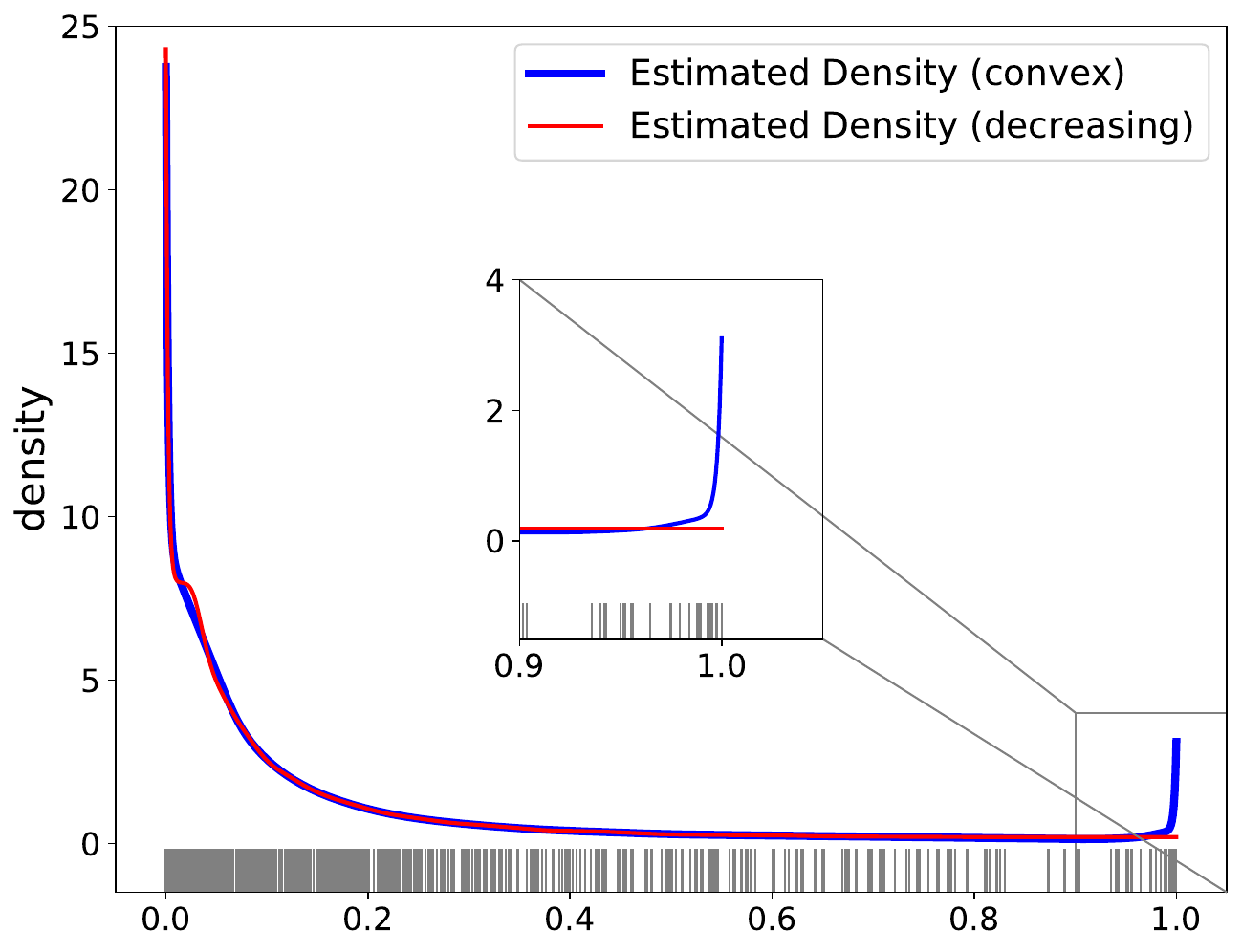}&
\includegraphics[width=0.5\textwidth]{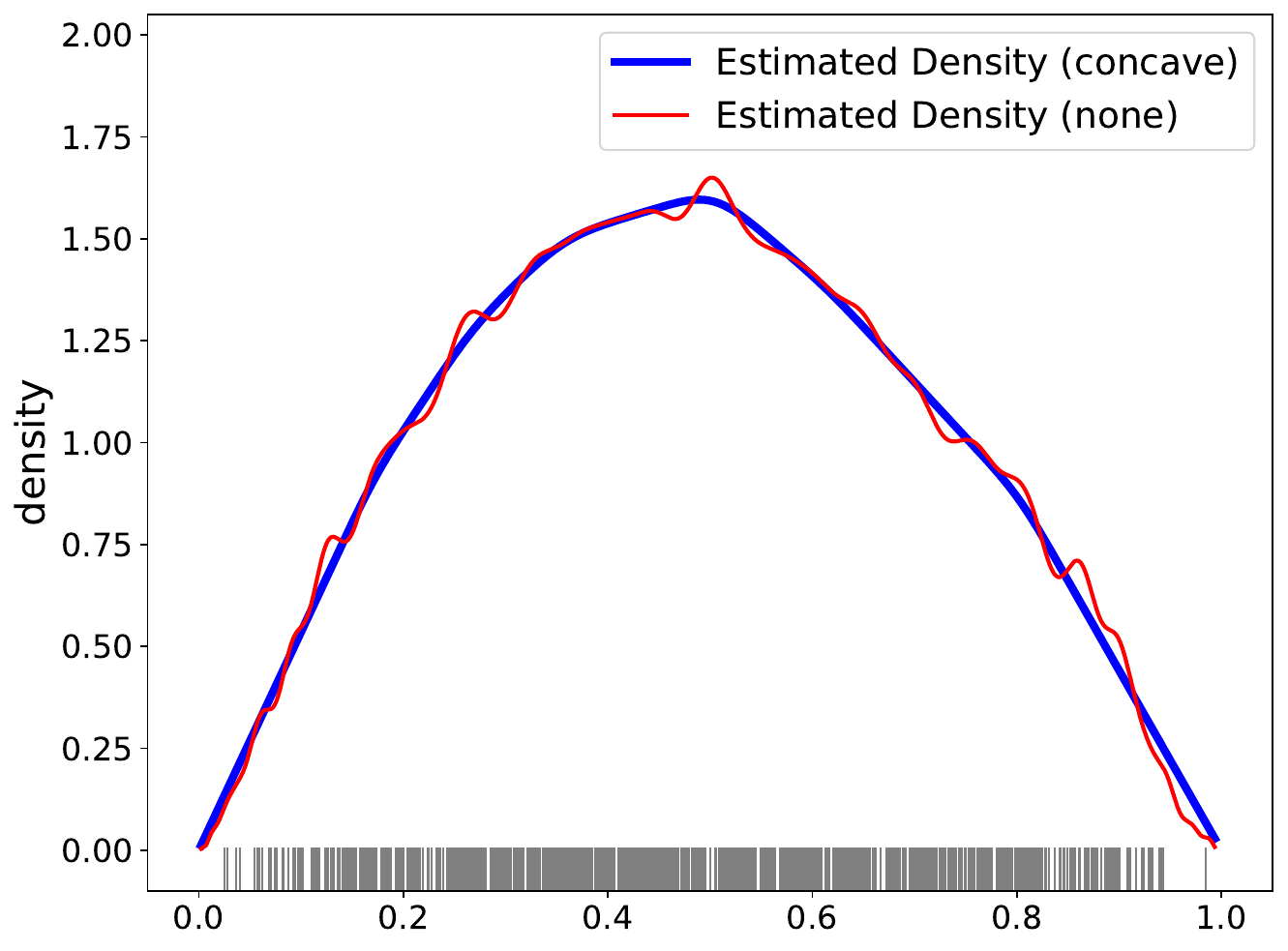}\\
(c) feature-3 & (d) feature-4
\end{tabular}}\caption{{\small Marginal probability density estimates 
of 4 covariates (all normalized between $[0,1]$) from the ACS dataset under different shape restrictions: decreasing, increasing, convex and concave as shown in the headings for each sub-figure. Estimators are obtained by solving~\eqref{problem1} with Bernstein bases with $M=500$ components. For (a), (b) the x-axis ranges have been truncated to aid visualization. }}         \label{fig:Census-ShapeRestrictions}
\end{figure}

Figure~\ref{fig:Census-ShapeRestrictions} shows a visualization of the estimated densities for 4 different covariates. Figure~\ref{fig:Census-ShapeRestrictions}~(a) compares our density estimate with Grenander's monotone density estimator~\cite{grenander1956theory}. We see that Grenander's estimator is not smooth whilst our estimator is smooth (due to the choice of the Bernstein polynomials as bases elements). Our estimator appears to be more expressive near $x=0$. 
Figure~\ref{fig:Census-ShapeRestrictions} (c) shows two of our estimators: one convex, and the other decreasing. The inset shows a zoomed-in version of the right tail showing the contrasting behavior between the convex and decreasing density estimates. Figure~\ref{fig:Census-ShapeRestrictions} (d) shows two estimators: one with no shape constraint~\eqref{crit-our} and the other with a concave shape constraint. The shape constrained density estimator appears to be less `wiggly' than the unconstrained variant---consistent with our observation in Figure~\ref{fig:Synthetic-compare1-shape}.

For the same dataset, Table~\ref{table: compare timing census} presents running times of our method versus Mosek for Problem~\eqref{problem1} using Bernstein bases elements and different shape constraints as shown in the table. We observe that our algorithm reaches a reasonable accuracy in runtimes that are smaller than Mosek. We see up to a 15X improvement in runtimes over Mosek. 

\begin{table}[ht!]
\centering
\scalebox{0.9}{\begin{tabular}{|c|c|c|c|c|}
\hline
Census Feature & Shape Constraint & $M$ & Mosek & Algorithm~\ref{algorithm: Self-concordant Cubic-regularized Newton Method} \\
& ($S$ in~\eqref{problem1})  & & time (s) & time (s)/ error \\ \hline
\multirow{3}{*}{feature-1} & \multirow{3}{*}{Decreasing} & 100 & 18.7 & 7.0/9.9e-5 \\ \cline{3-5}
 &  & 500 & 51.5 & 11.1/1e-4 \\ \cline{3-5} 
 &  & 1000 & 79.2 & 27.5/8.9e-5 \\ \hline
\multirow{3}{*}{feature-2} & \multirow{3}{*}{Increasing} & 100 & 31.8 & 6.4/2.0e-5 \\ \cline{3-5} 
 &  & 500 & 74.7 & 6.7/7.9e-5 \\ \cline{3-5} 
 &  & 1000 & 116.9 & 17.2/8.46e-5 \\ \hline
\multirow{3}{*}{feature-3} & \multirow{3}{*}{Convex} & 100 & 28.5 & 2.1/6.3e-5 \\ \cline{3-5} 
 &  & 500 & 65.3 & 7.7/8.8e-5 \\ \cline{3-5} 
 &  & 1000 & 102.7 & 20.9/9.4e-5 \\ \hline
\multirow{3}{*}{feature-4} & \multirow{3}{*}{Concave} & 100 & 23.2 & 1.15/6.7e-5 \\ \cline{3-5} 
 &  & 500 & 76.9 & 4.8/7.1e-5 \\ \cline{3-5} 
 &  & 1000 & 158.7 & 10.7/8.2e-5 \\ \hline
\end{tabular}}\caption{\small Density estimation with shape constraints for Problem~\eqref{problem1} on four ACS features from the Census Bureau (see Figure~\ref{fig:Census-ShapeRestrictions} for corresponding visualizations). Here, $N \approx 74,000$ and we take different values of $M$. We compare Mosek vs our method (Algorithm~\ref{algorithm: Self-concordant Cubic-regularized Newton Method}). We report running time (s) and solution accuracy i.e., relative error---reported as time/error.}
\label{table: compare timing census}
\end{table}

\section*{Acknowledgments}
Rahul Mazumder would like to thank Hanzhang Qin for helpful discussions on Bernstein polynomials during a class project in MIT course 15.097 in 2016.
The authors also thank Bodhisattva Sen for helpful discussions, 
%	This research is partially supported by ONR-N000141812298, NSF-IIS-1718258 and awards from IBM and Liberty Mutual Insurance.
and 
%	The authors would like to 
	thank the Associate Editor and anonymous referees for their comments that helped improve the paper. 
	This research is partially supported by ONR (N000142212665, N000141812298, N000142112841), NSF-IIS-1718258 and awards from IBM and Liberty Mutual Insurance.

\clearpage 
\bibliographystyle{plain}
{\small{\bibliography{main}}}

\clearpage 
\appendix

\section{Frank-Wolfe algorithm with away steps}\label{sec:append:AFW}

Below we present the details for Frank-Wolfe algorithm with away steps \cite{guelat1986some,lacoste2015global}. 
For notational convenience, we denote the objective $h_{f}(y, w^k)$ in~\eqref{update-yk} by
$$
G(z): = \Phi_f(z,w^k) + \frac{L_k}{6} 
\|  z -w^k \|_{\na^2 f(w^k)}^{3} \ .
$$

In addition, let $V(\cC)$ denote the set of vertices of the constraint set $\cC$ (which is a bounded convex polyhedron in our problem). 
In the AFW algorithm (see Algorithm~\ref{alg: AFW}), we start from some $z^0$ which is a vertex of $\cC$. 
In iteration $t$, $z^t$ can always be represented as a convex combination of at most $t$ vertices $V(z^t)\subset \cV(\cC)$, which we write as $z^t = \sum_{v\in V(z^t)} \lam_{v}(z^t) v$, where $\lam_{v}(z^t)>0$ for $v\in V(z^t)$, and $ \sum_{v\in V(z^t)} \lam_{v}(z^t) = 1.$ We keep track of the set of vertices $V(z^t)$ and the weight vector $\lam(z^t)$, and choose between the Frank-Wolfe step and away step (See Algorithm~\ref{alg: AFW}, Step~1). After a step is taken and we move to $z^{t+1}$, we update the set of vertices and weights to $V(z^{t+1})$ and $\lam(z^{t+1})$, respectively.
More precisely, we update $V(z^{t+1})$ and $\lam(z^{t+1})$ in the following way: For a Frank-Wolfe step, let $V(z^{t+1}) = \{s^t\}$ if $\al_t=1$; otherwise $V(z^{t+1}) = V(z^t) \cup \{s^t\}$. Also, we have $ \lam_{s^t}(z^{t+1}) := (1-\al_t) \lam_{s^t} (z^t) + \al_t $ and $\lam_v(z^{t+1}) = (1-\al_t) \lam_v (z^{t+1}) $ for $v\in V(z^t) \backslash \{s^t\}$. For an away step, let $V(z^{t+1}) = V(z^{t}) \backslash \{v^t\} $ if $\al_t= \al_{\max}$; otherwise $V(z^{t+1}) = V(z^t)$. Also, we have $ \lam_{v^t} (z^{t+1}) := (1+ \al_t) \lam_{v^t} (z^t) - \al_t $ and $ \lam_{v}(z^{t+1}) := (1+\al_t) \lam_v^{t} $ for $v\in V(z^t) \backslash \{v^t\}$.

\begin{algorithm}[htp]
	\caption{Away-step Frank-Wolfe method (AFW)~\cite{guelat1986some,lacoste2015global}}
	\label{alg: AFW}
	\begin{algorithmic}
		\STATE Starting from $z^{0}\in \cV(\cC)$.
		 Set $V(z^0) =\{ z^0 \}$ and $\lam_{z^0}(z^0) = 1$.
		\STATE For $t=0,1,2, \dots$:
		\STATE (1) 
		$s^{t}:= \arg\min_{z\in \cC}\la{\na G(z^t)},z\ra$,
		$v^t:= \argmax_{z\in V(z^t)} \la {\na G(z^t)},z \ra $.
		\STATE \qquad \textbf{if} $\la \na G(z^t), s^t-  z^t \ra < \la \na G(z^t), z^t - v^t \ra  $: (Frank-Wolfe step)
		\STATE \qquad \quad $d^t:= s^t -  z^t$; $\al_{\max} = 1$.
		\STATE \qquad \textbf{else}: (away step)
		\STATE \qquad \quad $d^t:=z^t - v^t$; $\al_{\max} = {\lam_{v^t}(z^t)  }/({1- \lam_{v^t}(z^t)})$
		\STATE (2)  $z^{t+1} = z^{t}+\alpha_{t}d^t$ ; 
		where $\alpha_{t} \in \argmin\limits_{\al\in[0,\al_{\max}]} G(z^t + \al d^t )$.
		\STATE (3) Update $V(z^{t+1})$ and $\lam(z^{t+1})$ (see text for details).
	\end{algorithmic}
\end{algorithm}

\section{Postponed Proofs}\label{all-proofs}

\subsection{Proofs on Computational Guarantees}\label{sec:append:computational-guarantees}

\subsubsection{Proof of Lemma~\ref{lemma: Boundedness of L_k}}
It suffices to show that if $L_k \ge   24N$, condition \eqref{condition1} must be satisfied. 
(If this is true, then noting that $\beta \in (1,2)$, we know $L_k \le \max\{48N, L_0\}  $).
For $i\in [N]$, recall that $B_i:= [B_{1i}, B_{2i}, ..., B_{Mi}]^\T \in \R^M$. Define
\begin{equation}
\mu_i(y,w) := \la B_i , w-y \ra / \la B_i ,w \ra.
\end{equation}
Then it holds that 
$$\la \na f(w), y-w \ra = \frac{1}{N} \sum_{i=1}^N \mu_i(y,w)$$ and 
\begin{equation*}
\na^2 f(w) [y-w]^2 =\frac{1}{N}\sum_{i=1}^N (y-w)^\T (\la B_i , w \ra)^{-2} B_i B_i^\T (y-w) = \frac{1}{N}\sum_{i=1}^N \mu_i(y,w)^2.  
\end{equation*}
As a result, recalling the definition of $\Phi_f$ in \eqref{def: Phi_f}, we have 
\begin{eqnarray}
&&\Phi_f(y,w) + (L/6) \| y-w \|_{\na^2 f(w)}^3 - f(w)
\nonumber\\
&=& \frac{1}{N}  \sum_{i=1}^N \Big\{ \mu_i(y,w)  + \frac{1}{2} \mu_i(y,w)^2  \Big\} + \frac{L}{6} \Big( \frac{1}{N} \sum_{i=1}^N \mu_i(y,w)^2  \Big)^{3/2} \nonumber 
\end{eqnarray}
for all $w,y$ and $L$. 
From Assumption \ref{assumption:  sub-problem solution}, we have
\begin{eqnarray}\label{ineq15}
\frac{1}{N}  \sum_{i=1}^N \Big\{ \mu_i(y^k,w^k)  + \frac{1}{2} \mu_i(y^k,w^k)^2  \Big\} + \frac{L_k}{6} \Big( \frac{1}{N} \sum_{i=1}^N \mu_i(y^k,w^k)^2  \Big)^{3/2} \le 0. 
\end{eqnarray}
Let $p:= \frac{1}{N} \sum_{i=1}^N \mu_i(y^k,w^k)$ \text{and} $q:= (\frac{1}{N} \sum_{i=1}^N \mu_i(y^k,w^k)^2)^{1/2},$ then \eqref{ineq15} is equivalent to $p + (1/2) q^2 + ({L_k}/{6})) q^3 \le 0$. 
By Jensen's inequality, it holds $p^2 \le q^2$, so we have 
\begin{equation}
({L_k}/{6}) q^3  \le 
({1}/{2}) q^2 + ({L_k}/6) q^3 \le 
-p \le q \ , \nonumber
\end{equation}
and thus $q \le (6/L_k)^{1/2}$. 
Recall that $q^2 = \nabla^2 f(w^k)[y^k - w^k]^2$, and making use of the definition of $\| \cdot \|_{F,x}$ (see~\eqref{defn-norm-F,x})  along with Lemma~\ref{lemma: f-self-concordant}, we have:
$$
\sqrt{N} q = \sqrt{ N \na^2 f(w^k) [y^k - w^k ]^2 } = \| y^k - w^k \|_{f,w^k} \ .
$$
By using the observation $q \le (6/L_k)^{1/2}$ and the assumption $L_k \ge 24 N$, we have $\sigma_k:= \| y^k - w^k \|_{f,w^k} = \sqrt{N} q  \le 1/2 $. 
We have the following chain of inequalities
\begin{eqnarray}
&&
\Big|   f(y^k) - f(w^k) - \la \na f(w^k), y^k-w^k  \ra - \frac{1}{2} \na^2 f(w^k) [y^k - w^k]^2   \Big| \nonumber\\
&\mathop{=}\limits^{(i)}&
\Big| \int_{[0,1]^3} t^2s \na^3 f(w^k+rst(y^k-w^k)) [y^k-w^k]^3 \text{d}s \ \rmd t \ \rmd r \Big| \nonumber\\
&\le&
\int_{[0,1]^3} t^2 s \lt|  \na^3 f(w^k+rst(y^k-w^k)) [y^k-w^k]^3 \rt| \rmd s \ \rmd t \ \rmd r  \nonumber\\
&\mathop{\le}\limits^{(ii)}&
2 \sqrt{N} \int_{[0,1]^3} t^2 s   \lt|  \na^2 f(w^k+rst(y^k-w^k)) [y^k-w^k]^2 \rt|^{3/2} \rmd s \ \rmd t \ \rmd r  \nonumber\\
&\mathop{\le}\limits^{(iii)}&
2 \sqrt{N} \int_{[0,1]^3} t^2 s \frac{1}{(1-rst\sigma_k)^3} \lt| \na^2 f(w^k) [y^k - w^k]^2 \rt|^{3/2} \rmd s \ \rmd t \ \rmd r \nonumber\\
&\mathop{\le}\limits^{(iv)}&
16 \sqrt{N} \lt| \na^2 f(w^k) [y^k - w^k]^2 \rt|^{3/2}  \int_{[0,1]^3} t^2 s \ \rmd s \ \rmd t \ \rmd r \nonumber\\
&\le&
3 \sqrt{N} \lt| \na^2 f(w^k) [y^k - w^k]^2 \rt|^{3/2} \nonumber
\end{eqnarray}
where $(i)$ is from Taylor formula (see Lemma \ref{lemma: Talor}), $(ii)$ uses the fact that $f$ is a $1/N$-self concordant function, $(iii)$ is from \eqref{eqn: hessian compare}, and $(iv)$ is from $\sigma_k \le 1/2$. As a result, using $L_k \ge   24N$, we have:
\begin{eqnarray}
f(y^k ) &\le& \Phi_{f} (y^k, w^k)
+ 3 \sqrt{N} \lt| \na^2 f(w^k) [y^k - w^k]^2 \rt|^{3/2} \nonumber\\
&\le&
\Phi_{f} (y^k, w^k) + ({L_k}/{6}) \lt| \na^2 f(w^k) [y^k - w^k]^2 \rt|^{3/2} \nonumber
\end{eqnarray}
and hence Condition \eqref{condition1} is satisfied for any non-negative $\ga_k$.

\subsubsection{Proof of Theorem~\ref{theorem: global convergence}}
From Step 3 of Algorithm \ref{algorithm: Self-concordant Cubic-regularized Newton Method} and Assumption \ref{assumption:  sub-problem solution}, we have 
\begin{eqnarray}\label{ineq-1}
f(w^{k+1}) &\le& f(y^k) + \rho_k \nonumber\\
&\le&
\Phi_{f} (y^k,w^k) + \frac{L_k}{6} \| y^k- w^k \|^3_{\na^2 f(w^k) } + \rho_k + \ga_k \nonumber\\
&\le&
\min_{y\in \cC} \Big\{ 	\Phi_{f} (y,w^k) + \frac{L_k}{6} \| y- w^k \|^3_{\na^2 f(w^k)}  \Big\}  + E_k
\end{eqnarray} 
where the first inequality is by \eqref{condition2}, the second inequality is by \eqref{condition1}, and the third inequality follows from~\eqref{update-yk} and the definitions of $\dt_k$ and $E_k$. 
Recall that $w^*$ is an optimal solution of \eqref{problem1}. 
For $\tau \in [0,1]$, let $w^k(\tau):= \tau w^* + (1-\tau ) w^k$. 
Restricting the minimization in \eqref{ineq-1} to this line segment joining $w^k$ and $w^*$, 
we have 
\begin{eqnarray}\label{ineq-1.5}
f(w^{k+1}) \le	\min_{\tau\in [0,1]} \Big\{
\Phi_{f} (w^k(\tau),w^k) +\frac{L_k\tau^3}{6} \| w^*- w^k \|^3_{\na^2 f(w^k) }
\Big\} + E_k. 
\end{eqnarray}
By Taylor's expansion in Lemma \ref{lemma: Talor}, we have 
\begin{equation}\label{long1}
\begin{aligned}
&~~~~\Phi_{f} (w^k(\tau),w^k)  \\
& \mathop{=}f(w^k(\tau)) + 
\int_{[0,1]^3} \frac{2}{N} \sum_{i=1}^N
t^2s \lt(   \frac{\la B_i, w^k(\tau) - w^k \ra}{ \la B_i , w^k + str (w^k(\tau) - w^k) \ra  }    \rt)^3 \rmd s \ \rmd t \ \rmd r.
\end{aligned}
\end{equation}
Note that 
\begin{equation}\label{long2}
\begin{aligned}
&~~~~\frac{2}{N} \sum_{i=1}^N \lt(   \frac{\la B_i, w^k(\tau) - w^k \ra}{ \la B_i , w^k + str (w^k(\tau) - w^k) \ra  }    \rt)^3 \\
&= \frac{2}{N} \sum_{i=1}^N \tau^3 \lt(   \frac{\la B_i, w^* - w^k \ra}{ (1-str\tau) ( \la B_i, w^k \ra) + str\tau ( \la B_i , w^* \ra )  }    \rt)^3 \\
& \le 
\frac{2}{N} \sum_{i=1}^N \tau^3
\max \Big\{ 
\Big| \frac{\la B_i, w^* - w^k \ra }{  \la B_i , w^k \ra } \Big|^3, 
\Big| \frac{\la B_i, w^* - w^k \ra }{  \la B_i , w^* \ra } \Big|^3
\Big\} \\
&\le 2 C_1 \tau^3
\end{aligned}
\end{equation}
where, the first inequality is because $a /((1-\al)b + \al c ) \le \max \{a/b, a/c \}  $ for all $a,b,c \ge 0$ and $\al\in [0,1]$, and 
the last inequality is from the definition of $C_1$ in \eqref{def: C}
and the fact that $w^k \in X^0$ for all $k\ge 1$. 
By \eqref{long1} and \eqref{long2} we have
\begin{equation}\label{ineq-2}
\Phi_{f} (w^k(\tau),w^k) \le f(w^k(\tau)) + 2 C_1 \tau^3 	\int_{[0,1]^3}  t^2 s \ \rmd t \ \rmd s \ \rmd r = 
f(w^k(\tau)) +  \frac{ C_1 \tau^3}{3} \ .
\end{equation}
On the other hand, note that $ \| w^*- w^k \|^2_{\na^2 f(w^k) } = \frac{1}{N} \sum_{i=1}^N  \la B_i, w^* - w^k \ra^2  / \la B_i, w^k \ra^2, $ so
\begin{eqnarray}\label{ineq-3}
\| w^*- w^k \|^3_{\na^2 f(w^k) } 
\le 
\frac{1}{N} \sum_{i=1}^N \Big|  \frac{\la B_i, w^* - w^k \ra}{ \la B_i, w^k \ra } \Big|^3 \le C_1 \ ,
\end{eqnarray}
where we use Jensen's inequality for the first inequality; and~\eqref{def: C} for the second inequality. Combining inequalities \eqref{ineq-1.5}, \eqref{ineq-2} and \eqref{ineq-3}, we obtain
\begin{eqnarray}
f(w^{k+1}) &\le & f(w^k(\tau)) + \frac{C_1\tau^3}{3} + \frac{1}{6} \sup_{k\ge 0 } \{L_k\} C_1\tau^3 + E_k \nonumber\\
&\le&
f(w^k(\tau)) + C \tau^3 +E_k \label{proof-wtau-line-2} 
\end{eqnarray}
where, the second inequality uses the definition of $C$ in~\eqref{def: C}. 
Using the convexity of $\tau \mapsto f(w^k(\tau))$ in line~\eqref{proof-wtau-line-2}, we obtain 
\begin{equation}\label{proof-line-xyz-1}
f(w^{k+1}) \le 
(1-\tau) f(w^k) + \tau f^* + C \tau^3 + E_k
\end{equation}
which holds for all $\tau\in [0,1]$.
Taking $\tau_k = 3/(k+3)$ above, we get 
\begin{equation} \label{proof-line-xyz-111}
f(w^{k+1}) - f^* \le (1-\tau_k) (f(w^k) - f^*) + \frac{27C}{(k+3)^3} + E_k \ . 
\end{equation}
Let $T_0 = 1$ and $T_k = \prod_{i=1}^k(1-\tau_i) $, then we have 
\begin{equation}
T_k = \frac{6}{(k+1)(k+2)(k+3)}\quad \text{for }k\ge 1 \ . \nonumber
\end{equation}
It follows from~\eqref{proof-line-xyz-111} and the definition of $T_{k}$ that
\begin{eqnarray}
\frac{1}{T_k} (f(w^{k+1}) - f^* ) \le \frac{1}{T_{k-1}} (f(w^k) - f^*) + \frac{27C}{(k+3)^3 T_k} + \frac{E_k}{T_k} \ .  \nonumber
\end{eqnarray}
Summing the above over $k$ yields
\begin{equation}\label{proof-line-xyz1-1}
\frac{1}{T_k} (f(w^{k+1}) - f^* ) \le
\frac{1}{T_0} (f(w^{1}) - f^* )  + \sum_{i=1}^k \frac{27C}{(i+3)^3 T_i } + \sum_{i=1}^k \frac{E_i}{T_i} \ . 
\end{equation}
Note that 
$$f(w^1) - f^* \le f(w^0 ) - f^* \le C,~~T_0 = 1~~~\text{and}~~~27/(T_i (i+3)^3)  \le 9/2,$$
which when used in~\eqref{proof-line-xyz1-1}, leads to
\begin{equation}
\frac{1}{T_k} (f(w^{k+1}) - f^* ) \le 
C + \frac{9}{2} kC +  \sum_{i=1}^k \frac{E_i}{T_i} 
\le
5kC +  \sum_{i=1}^k \frac{E_i}{T_i}  \ \nonumber
\end{equation}
where the second inequality uses the assumption $k\ge2$.
As a result, 
\begin{equation}
\begin{aligned}
f(w^{k+1}) - f^* &\le \frac{30C}{(k+2)(k+3)} + \frac{1}{(k+1)(k+2)(k+3)} \sum_{i=1}^k (i+1)(i+2)(i+3) E_i  \\
&\le \frac{30C}{(k+2)^2} + \frac{1}{(k+1)^3} \sum_{i=1}^k (i+3)^3 E_i 
\end{aligned} \nonumber
\end{equation}
which completes the proof.

\subsubsection{Proof of Theorem~\ref{theorem: local convergence}}
In this proof, for notational convenience, we use the shorthand notations: 
$$H_k := \na^2 f(w^k), ~~\eta_k:= \| w^{k} - w^* \|_{H_k}~~\text{and}~~\eta_{k}':= \| w^{k+1} - w^* \|_{H_k}.$$ 
Suppose for some $k\ge K$, we have:
\begin{eqnarray}\label{local-condition}
(6\sqrt{N}+2\bar L) \| w^k - w^* \|_{H_*} \le 1 \ .
\end{eqnarray}
It suffices to prove that 
\begin{equation}\label{k-conclusion}
\| w^{k+1} - w^* \|_{H_*} \le  (6\sqrt{N} + 2\bar L ) \| w^{k} - w^* \|_{H_*}^2 \ . 
\end{equation}
This is because, by \eqref{local-condition} and \eqref{k-conclusion} it follows that $(6\sqrt{N}+2\bar L) \| w^{k+1} - w^* \|_{H_*} \le 1$---we can then apply the same arguments for $k+1$ in place of $k$ and prove Theorem \ref{theorem: local convergence} by induction. 

Below we prove \eqref{k-conclusion} assuming that~\eqref{local-condition} holds true. 
First, note that by condition~\eqref{local-condition}, we have
\begin{equation}\label{eta_k <= 1/4}
\| w^k - w^* \|_{f,w^*} = \sqrt{N} \| w^k - w^* \|_{H_*} \le 1/6  \ .
\end{equation}
Hence by \eqref{eqn: hessian compare} we have 
\begin{eqnarray}\label{compare-hess}
(25/36) H_* \preceq H_k \preceq (36/25) H_* \ .
\end{eqnarray}

Let $h$ be the indicator function of the set $\cC$, that is, $h(w) = 0$ for $w\in \cC$ and $h(w) = \infty$ for $w\notin \cC$. Since by our assumption, 
$y^{k}=w^{k+1}$ is a minimizer for~\eqref{update-yk}, we have
\begin{eqnarray}
\na f(w^k)  + \mu_k \na^2 f(w^k) (w^{k+1} - w^k) + v^{k+1} = 0 \ , \nonumber
\end{eqnarray}
where 
$v^{k+1} \in \pa h(w^{k+1})$ and $\mu_k:=1+ \frac{L_k}{2} \| w^{k+1} - w^k \|_{H_k} $. 
On the other hand, since $w^*$ is the optimal solution to~\eqref{problem1}, 
there exists $v^* \in \pa h(w^*)$ such that $\na f(w^*) + v^* = 0$---this leads to
\begin{eqnarray}\label{ineq11}
\na f(w^k) - \na f(w^*)  + \mu_k \na^2 f(w^k) (w^{k+1} - w^k) + v^{k+1} - v^* = 0 \ .
\end{eqnarray}

Multiplying both sides of~\eqref{ineq11} by $w^{k+1} -w^*$, and in view of $(w^{k+1} - w^*)^\T (v^{k+1} - v^*) \ge 0$ (since $h$ is a convex function), we have
\begin{eqnarray}
0 &\ge &
(w^{k+1} - w^*)^\T \Big(   
\na f(w^k) - \na f(w^*) + \mu_k \na^2 f(w^k) (w^{k+1} - w^k)
\Big) \nonumber\\
&=&
(w^{k+1} - w^*)^\T \Big(  
\na f(w^k) - \na f(w^*) - \mu_k \na^2 f(w^k) (w^{k} - w^*)
\Big) + \mu_k \| w^{k+1} - w^* \|^2_{H_k}  . \nonumber
\end{eqnarray}
Recall that $ \eta_{k}' = \| w^{k+1} - w^* \|_{H_k}$.  
Hence from the above, we get:
\begin{eqnarray}
\mu_k (\eta_{k}')^2 &\le& ( w^{k+1} - w^* )^{\T}  \lt(
\na f(w^*) - \na f(w^k) - \mu_k \na^2 f(w^k) (w^* - w^k)
\rt) \nonumber\\
&\le&
\eta_{k}' 
\|
H_k^{-1/2} \lt(  \na f(w^*) - \na f(w^k) - \mu_k \na^2 f(w^k) (w^* - w^k)   \rt)
\|_2 \ , \nonumber
\end{eqnarray}
and hence
\begin{equation}\label{ineq12}
\begin{aligned}
\mu_k \eta_{k}' &\le 
\|
H_k^{-1/2} \lt(  \na f(w^*) - \na f(w^k) - \mu_k \na^2 f(w^k) (w^* - w^k)   \rt)
\|_2 \\ 
&\le
\|
H_k^{-1/2} \lt(  \na f(w^*) - \na f(w^k) - \na^2 f(w^k) (w^* - w^k)   \rt)
\|_2 \\
&~~~~~~+ 
(L_k/2) \| w^{k+1} - w^k \|_{H_k } \| w^k - w^* \|_{H_k} \ ,
\end{aligned}
\end{equation}
where the second inequality in~\eqref{ineq12} makes use of the definition of $\mu_k=1+ \frac{L_k}{2} \| w^{k+1} - w^k \|_{H_k}$. 
Let us denote
\begin{eqnarray*}
	&J_1:= \|
	H_k^{-1/2} \big(  \na f(w^*) - \na f(w^k) - \na^2 f(w^k) (w^* - w^k)   \big)
	\|_2  ~~~&~~~~~\text{and}\nonumber\\
	&J_2:=(L_k/2) \| w^{k+1} - w^k \|_{H_k } \| w^k - w^* \|_{H_k}.~~~~~~&
\end{eqnarray*}
Then~\eqref{ineq12} reduces to 
\begin{eqnarray}\label{bound-J1J2}
\mu_k \eta_{k}' \le J_1 + J_2 \ .
\end{eqnarray}
By the definition of $J_1 $ one has
\begin{eqnarray}\label{eq-J1}
J_1 = \Big\| \int_{0}^1 
H_k^{-1/2}  \big(  \na^2 f(w^k(t) )  - \na^2 f(w^k) \big) (w^* - w^k)  \rmd t
\Big\|_2  \ , 
\end{eqnarray}
where $w^k(t) = w^k + t (w^* - w^k)$.

Let $$r:= \| w^k - w^*\|_{f,w^k} = \sqrt{N} \| w^k - w^*\|_{H_k}.$$ 
By \eqref{compare-hess} we have 
\begin{equation}\label{bound-r}
r \le (6/5) \sqrt{N} \| w^k - w^*\|_{H_*} = (6/5) \| w^k - w^*\|_{f,w^*} \le 1/5 .
\end{equation}
where, the second inequality is by \eqref{eta_k <= 1/4}.

For any $t\in [0,1]$, by the definition of $r$ and~\eqref{bound-r},
we have 
$$\| w^k(t) - w^k \|_{f,w^k} = t  \| w^* - w^k \|_{f,w^k} = tr \leq 1/5<1.$$ 
Therefore, 
using \eqref{eqn: hessian compare}, for $t\in [0,1]$, we obtain
\begin{equation}\label{ineq13}
\na^2 f(w^k(t)) - \na^2 f(w^k) 
\preceq  \lt( (1-tr)^{-2} - 1 \rt) \na^2 f(w^k) \preceq 3 tr \na^2 f(w^k) \ ,
\end{equation}
where the last inequality is by $tr\in [0,1/5]$ and Lemma \ref{lemma: basic3}.  
Similarly, we can obtain:
\begin{equation}\label{ineq14}
\na^2 f(w^k (t)) - \na^2 f(w^k)  
\succeq  	\lt( (1-tr)^2 -1  \rt) \na^2 f(w^k)  \succeq 
-3 tr \na^2 f(w^k).
\end{equation}
Combining \eqref{ineq13} and \eqref{ineq14}, and using Lemma \ref{lemma: matrix inequality} with $\bar A =\na^2 f(w^k (t)) - \na^2 f(w^k)  $ and $\bar B = 3tr \na^2 f(w^k) = 3tr H_k$, we  have $\bar A\bar B^{-1}\bar A \preceq \bar B$, which implies: 
\begin{eqnarray}\label{key-1}
\| 
H_k^{-1/2}  \big(  \na^2 f(w^k(t))  - \na^2 f(w^k) \big) (w^* - w^k)  \|_2 \le 3tr\| w^* - w^k \|_{H_k}. 
\end{eqnarray}
Note that $r = \| w^k - w^*\|_{f,w^k} = \sqrt{N} \| w^k - w^* \|_{H_k} = \sqrt{N} \eta_k$ and $\| w^k - w^* \|_{H_k} =  \eta_k $, so \eqref{key-1}
is equivalent to
\begin{equation}\label{key0}
\| 
H_k^{-1/2}  \big(  \na^2 f(w^k(t))  - \na^2 f(w^k) \big) (w^* - w^k)  \|_2  \le 3t\sqrt{N}\eta_k^2.
\end{equation}
By~\eqref{eq-J1} and the convexity of the norm $\| \cdot\|_2$, one has
\begin{eqnarray}\label{bound-J1}
J_1 \le
\int_0^1 \| 
H_k^{-1/2}  \big(  \na^2 f(w^k(t))  - \na^2 f(w^k) \big) (w^* - w^k)  \|_2 \rmd t 
\le (3/2)\sqrt{N} \eta_k^2 
\end{eqnarray}
where, the second inequality is by~\eqref{key0} and $\int_{0}^1 t \rmd t = 1/2$.
On the other hand, 
\begin{eqnarray}\label{bound-J2}
J_2 = ({L_k}/{2}) \| w^{k+1} - w^k \|_{H_k } \| w^k - w^* \|_{H_k} \le ({\bar L}/{2}) (\eta_k + \eta_{k}') \eta_k
\end{eqnarray}
where, the inequality in~\eqref{bound-J2} follows by noting that $L_k \le \bar L $ and 
$$\| w^{k+1} - w^k \|_{H_k }\le \| w^{k+1} - w^* \|_{H_k } + \| w^{k} - w^* \|_{H_k } = \eta_k + \eta_k' .$$ 
Combining \eqref{bound-J1J2}, \eqref{bound-J1} and \eqref{bound-J2}, we have 
\begin{eqnarray}\label{ineq--2}
\mu_k \eta_{k}' \le (3/2)\sqrt{N} \eta_k^2 + ({\bar L}/{2}) (\eta_k + \eta_{k}') \eta_k \ .
\end{eqnarray}
By \eqref{compare-hess} and assumption \eqref{local-condition} we have 
\begin{eqnarray}
\bar L \eta_k = \bar L \| w^k-w^* \|_{H_k} \le (6/5) \bar L \| w^k-w^* \|_{H_*} \le 1, \ \nonumber
\end{eqnarray}
which implies $1/2 - \bar{L}\eta_{k}/2 \geq 0$. 
Hence we have 
\begin{equation}\label{key1}
\begin{aligned}
\eta_{k}' /2 \le& \eta_{k}'/2 +  \eta_{k}'/2 - \bar L \eta_k \eta_{k}'/2 \\
=& \eta_{k}' - \bar L \eta_k \eta_{k}'/2 \\
\le& 
\mu_k \eta_{k}' - \bar L \eta_k \eta_{k}'/2 \\
\le&
((3/2) \sqrt{N} + {\bar L}/{2}) \eta_k^2,
\end{aligned}
\end{equation}
where the second inequality is because
$\mu_k \ge 1 $, and the third inequality is by
\eqref{ineq--2}. From \eqref{key1} and the definitions of $\eta_k$ and $\eta_{k}'$, we have
\begin{eqnarray}\label{key2}
\| w^{k+1} - w^* \|_{H_k} \le (3\sqrt{N} + \bar L ) \| w^{k} - w^* \|_{H_k}^2 \ . 
\end{eqnarray}
From \eqref{key2} and 
using \eqref{compare-hess}, we have 
\begin{equation}
\| w^{k+1} - w^* \|_{H_*} \le (6/5)^3 (3\sqrt{N} + \bar L ) \| w^{k} - w^* \|_{H_*}^2  \le (6\sqrt{N} + 2\bar L) \| w^{k} - w^* \|_{H_*}^2 \  \nonumber
\end{equation}
which completes the proof.

\subsection{Proofs: computing the LP oracles}\label{sec:append:LP-oracles}

\subsubsection{Proof of Proposition~\ref{proof-proposition-4.1}}

	If $w \in \cC^-$, define $y\in \R^M$ by $y_M = M w_M$ and $y_i = i (w_i - w_{i+1})$ for $i\in [M-1]$. Then equivalently we have $w = Uy$. Moreover, 
	by the definition of $\cC^-$, we have $y\ge 0$ and 
	$
	1_M^\T y = Mw_M + \sum_{i=1}^{M-1} i(w_i - w_{i+1}) = \sum_{j=1}^M w_j  =1.
	$
	Hence $y\in \Dt_M$, so we have proved $\cC^- \subseteq U(\Dt_M)$.

	On the other hand, suppose $y\in \Dt_M$, and let $w = Uy$. Then 
	$
	1_M^\T w  = 1_M^\T U y = 1_M^\T y = 1
	$;
	additionally, we have $w\ge 0$ and $ w_1 \ge w_2 \ge \cdots \ge w_M $. Hence $w=Uy \in \cC^-$, and $\cC^- = U(\Dt_M)$.

\subsubsection{Proof of Proposition~\ref{proposition:4point3}}

	We divide all the inequality constraints in the set $\cC^{\wedge}$, defined in~\eqref{C: concavity}, into two groups:
	$$
	(g1): ~w_i \ge 0 ,~ i\in[M]; \qquad (g2):  ~ 2w_i \ge w_{i-1}+ w_{i+1} ,~ 2\le i \le M-1.
	$$	
	Suppose $\bar w$ is a vertex of $\cC^{\wedge}$. 
By the definition of a vertex, $\bar w$ is given by a solution to a system of $M$ independent equations.
	Since there is one equality constraint $1_M^\T w = 1$, there must exist $M-1$ independent inequality constraints from (g1) and (g2), for which equality holds. Since (g2) contains $M-2$ constraints, there must exist some $i\in [M]$ such that $\bar w_i = 0$. By the concavity constraints, it is easy to see that $\#\{i: \bar w_i = 0\} \le 2$, since otherwise, all coordinates will be zero. 
	
	If $\#\{i: \bar w_i = 0\} = 1$, then either $\bar w_1 = 0$ or $\bar w_M = 0$, and
	all constraints in (g2) will be active---this implies $\bar w$ is linear. Combined with the fact that $1_M^\T \bar w = 1 $, it follows that $\bar w$ must be either $v_1$ or $v_M$. 
	
	If $\#\{i: \bar w_i = 0\} = 2$, then we must have $\bar w_1 = 0$ and $\bar w_M = 0$, and the $M-3$ constraints in (g2) will be active at $\bar w$. As a result, $\bar w$ is piecewise linear with two pieces. 
	Combined with the constraint that $1_M^\T \bar w = 1 $, we know that $\bar w$ must be one of the points $v_2,\ldots,v_{M-1}$. 
	
	Finally, it is easy to check that any point among $v_1,\ldots,v_M$ satisfies $M-1$ active inequality constraints, hence it is a vertex. This completes the proof of this proposition.

\subsubsection{Proof of Proposition~\ref{proposition: vectices concavity}}

	We divide all the inequality constraints in the set $ \cC^{{\vee}}$, defined in
	\eqref{C: convexity}, into two groups:
	$$
	(g1): ~w_i \ge 0 ,~ i\in[M]; \qquad (g2):  ~ 2w_i \le w_{i-1}+ w_{i+1} ,~ 2\le i \le M-1.
	$$
	Suppose $\bar w$ is a vertex of $ \cC^{{\vee}}$.  
	Since there is one equality constraint $1_M^\T w = 1$, a vertex must have $M-1$ independent active inequality constraints that are also independent to the equality constraint. Since (g2) has $M-2$ constraints, there must exist some $i\in [M]$ such that $\bar w_i = 0$. 
	By the convexity constraints in (g2), it can be seen that there exists a pair $(i_{1}, i_{2})$ with $1\le i_1 \le i_2 \le M$ such that 
	$$
	\bar w_i = 0 ~ ~{\rm for} ~ i_1 \le i \le i_2;~~~~~\text{and}~~~~~ \bar w_i >0 ~ ~{\rm for}  ~ 1 \leq i<i_1, ~\text{or}~M \geq i>i_2. 
	$$
	We claim that either $i_1 = 1$ or $i_2 = M$. Otherwise, suppose $1<i_1 \le i_2 <M$, then by the definition of $i_1$ and $i_2$, there are $i_2-i_1+1$ active constraints in (g1), and there are at most $(i_1-2) + (M-i_2-1)$ active constraints in (g2) that are independent to those 
	active constraints in (g1)---in total, there are at most $M-2$ independent inequality constraints that are active, which is not possible for a vertex $\bar w$.
	
	If $i_1 = 1$, then there are $i_2$ active constraints in (g1). For $\bar w$ to be a vertex, we need at least $ M-1-i_2 $  active constraints in (g2) that are independent to the
	active constraints in (g1). By the definition of $i_2$, we must have $2w_{i_2} < w_{i_2-1} + w_{i_2+1}$, so we should have $2w_{i} = w_{i-1} + w_{i+1}$ for all $i_2+1 \le i \le M-1$ (in order that there are $ M-1-i_2 $  active constraints in (g2)). In this case, using the equality constraint $1_M^\T w = 1$, we can see that $\bar w$ is of the form 
	$$
	 a_k^{-1}   (0^\T_{M-k}, 1, 2, \dots , k)    
	$$
	for some $1 \le k \le M$. On the other hand, we can verify that any point of the above form satisfies $M-1$ active inequality constraints, and is hence a vertex.
	
	Similarly, if $i_2 = M$, we can show that a vertex $\bar w$ must be of the form 
	$$
	a_k^{-1}  (k, k-1,  \dots , 2,1, 0^\T_{M-k}) 
	$$
	for some $1 \le k \le M$.

\subsubsection{Proof of Proposition~\ref{proposition: concave increasing restriction}}

	Note that under the constraints $ 2w_{i} \ge w_{i-1} + w_{i+1} $ for $2\le i \le M-1$, the single inequality $ w_{M-1} \le w_M $ will imply $ w_1 \le \cdots \le w_M $. As a result, we can rewrite $\cC^{\wedge  +}$ as follows:
		\begin{eqnarray}\label{C: concave and increasing2}
		\cC^{\wedge  +} = \lt\{    w \in \R^M :~ 1_M^\T w = 1, ~w \ge 0, ~   w \text{ is concave}, ~~ w_{M-1}\le w_M \rt\}  .  \nonumber
		\end{eqnarray}
	In the following, we separate all the inequality constraints above into 3 groups:
	$$
	(g1): ~w_i \ge 0 ,~ i\in[M]; ~~ (g2):  ~ 2w_i \ge w_{i-1}+ w_{i+1} ,~ 2\le i \le M-1;~~
	(g3): w_{M-1} \le w_M.
	$$
	Let $\bar w$ be a vertex of $\cC^{\wedge  +}$. 
	We first claim that for all $2\le i \le M$, we must have $\bar w_i > 0$. 
	Suppose (for the sake of contradiction), for some $2\le i\le M$ it holds that $ \bar w_i  =0 $, then by the increasing constraints and $\bar w \ge0$, we have $ \bar w_j = 0 $ for all $1\le j\le i$. Moreover, by the concavity constraints, we have $ \bar w_k = 0 $ for all $i\le k \le M$. As a result, $\bar w = 0$, which is a contradiction to $1_M^\T \bar w = 1$.

	 When $\bar w_1$ is also positive, in order that there are $M$ independent active constraints in total, all constraints in (g2) and (g3) should be active, and hence $\bar w_i = \bar w_j$ for all $1\le i,j \le M$. Combining with $1_M^\T \bar w = 1$, we have $\bar w = v_1$. 
	 When $\bar w_1 = 0$ and the constraint in (g3) is not active, 
	 in order that there are $M$ independent active constraints in total, all constraints in (g2) should be active. Combining with $1_M^\T \bar w = 1$, we obtain $\bar w = v_M$. 
	 When $\bar w_1 = 0$ and the constraint in (g3) is active, in order that there are $M$ independent active constraints in total, there are at least $M-3$ constraints in (g2) being active, which means $\bar w$ is piecewise linear with two pieces, and is a constant on the second piece (since $w_{M-1} = w_M$). 
	 Combining with $1_M^\T \bar w = 1$, we know that $\bar w$ must be among $\{v_2,...,v_{M-1}\}$. 
	 
	 On the other hand, it is easy to check that $v_1,...,v_M$ indeed satisfy $M$ independent active constraints, hence they are indeed vertices of $\cC^{\wedge  +}$.

\subsubsection{Proof of Proposition~\ref{proposition: convex increasing restriction}}

	Note that under the convexity constraints $ 2w_{i} \le w_{i-1} + w_{i+1} $ for $2\le i \le M-1$, the single inequality $ w_{1} \le w_2 $ will imply $ w_1 \le \cdots \le w_M $. As a result, we can rewrite $\cC^{\vee +}$ as follows:
	\begin{eqnarray}\label{C: convex and increasing2}
	\cC^{\vee +} = \lt\{    w \in \R^M :~ 1_M^\T w = 1, ~w \ge 0, ~   w \text{ is convex}, ~~ w_{1}\le w_2 \rt\}. \nonumber
	\end{eqnarray}
	In the following, we separate all the inequality constraints above into 3 groups:
	$$
	(g1): ~w_i \ge 0 ,~ i\in[M]; \quad (g2):  ~ 2w_i \le w_{i-1}+ w_{i+1} ,~ 2\le i \le M-1;\quad 
	(g3): w_{1} \le w_2.
	$$
	Now let $\bar w$ be a vertex of $	\cC^{\vee +}$. If $\bar w_1 >0$, then by the increasing constraints, $\bar w_i >0$ for all $i\in [M]$, and all constraints in (g1) are inactive. In order that there are $M$ independent active constraints at $\bar w$, all constraints in (g2) and (g3) should be active, and hence $\bar w_i = \bar w_j$ for all $1\le i,j\le M$. Since also $1_M^\T \bar w = 1$, we have $\bar w = v_M$.

If $\bar w_1 = 0$ and $\bar w_2 > \bar w_1$, in order that there are $M$ independent active constraints at $\bar w$, all constraints in (g2) should be active, hence $\bar w$ is linear. Combined with $ 1_M^\T \bar w = 1 $, we have $\bar w = v_{M-1}$. If $\bar w_1 = 0$ and $\bar w_1 = \bar w_2$, in order that there are $M$ independent active constraints at $\bar w$, there are at least $M-3$ active constraints in (g2), which means $\bar w$ is piecewise linear, and is a constant on the first piece (since $\bar w_1 = \bar w_2$). Combined with $ 1_M^\T \bar w = 1 $, we know that $\bar w$ must be one of $v_1, ..., v_{M-2}$. 
	
		 On the other hand, it is easy to check that $v_1,...,v_M$ indeed satisfy $M$ independent active constraints, and are vertices of $\cC^{\vee +}$.

\subsubsection{Proof of Proposition~\ref{proposition-unimodality-formula}}

	In the proof, we use the notation $\cC^{um}_k(M) $ for $\cC^{um}_k$ to highlight the dependence on $M$.

	We first show that for all $k \in [M]$, if $\wtd w$ is a vertex of $\cC^{um}_k(M)$ satisfying $\wtd w_i >0$ for all $i\in [M]$, then $\wtd w = (1/M) 1_M$. Indeed, since $\wtd w$ is a vertex of $\cC^{um}_k(M)$, there are $M$ independent active constraints at $\wtd w$. Since there is only one equality constraint in the definition of $\cC^{um}_k(M)$, there are at least $M-1$ independent inequality constraints for which equality holds at $\wtd w$. We separate all the inequality constraints in $\cC^{um}_k(M)$ into 3 groups:
	$$
	(g1): w_i \ge 0 ,~ i\in[M]; ~~ (g2):    w_{i-1} \le w_i  ,~ 2\le i \le k;~~
	(g3): w_i \ge w_{i+1} , ~ k \le i \le M-1.
	$$
	Since $\wtd w_i>0$ for all $i\in [M]$, no constraints in (g1) are active. Since there are $(k-1) + (M-k) = M-1$ constraints in (g2) and (g3), all the constraints in (g2) and (g3) are active at $\wtd w$. This implies $\wtd w = (1/M) 1_M$. 
	
	Below we prove the main conclusion. 
	Let $\bar w$ be a vertex of $\cC^{um}_k(M)$. By the constraints in $\cC^{um}_k(M)$, $\bar w$ is a unimodal vector with all coordinates being non-negative, so
	there exist $k_1 , k_2$ satisfying $1\le k_1 \le k \le k_2\le M $ such that $\bar w_{i} >0$ for all $k_1\le i \le k_2$ and $\bar w_i = 0$ for  all $i<k_1$ or $i>k_2$. Let $\bar u \in \R^{k_2-k_1+1}$ be the subvector of $\bar w$ restricted to the positive entries --- that is, 
	 $\bar u_i = \bar w_{i+k_1-1}$ for $i\in [k_2-k_1+1]$. Then $\bar u$ is a vertex of $ \cC^{um}_{ k-k_1+1} (k_2 - k_1 +1)$ (Since otherwise this is a contradiction to our assumption that $\bar w$ is a vertex of $\cC^{um}_k(M)$).
From the discussion in the earlier paragraph, we have $\bar u_i = 1/ (k_2 - k_1 +1)$ for all $i\in [k_2 - k_1 +1]$. This implies $\bar w = v^{k_1,k_2}$. 

On the other hand, it is easy to check that for all $1\le k_1 \le k \le k_2 \le M$,  the vector $v^{k_1,k_2}$ satisfies $M$ independent active constraints, so it is indeed a vertex of $\cC^{um}_k(M)$. This completes the proof.

\subsection{Proofs from Section~\ref{sec: approximation-guarantees}}

\subsubsection{Proof of Proposition~\ref{prop: optimality gap}}

	It is immediate from the definition that $p^* \le \hat p_\cG$. On the other hand, define 
	$
	\wtd \nu_i :=    {N\hat \nu_i}/{\Gamma} 
	$
	for $i\in [N]$.
	Then by our assumption on $\Gamma$, it holds
	$$
	\max_{\mu \in \R} \sum_{i=1}^N \wtd \nu_i \varphi (X_i - \mu) = 
	\frac{N}{\Gamma}\max_{\mu \in \R} \sum_{i=1}^N \hat \nu_i \varphi (X_i - \mu) \le N \ ,
	$$
    hence $\wtd \nu$ is a feasible solution of the (infinite) dual problem \eqref{infinite-constraints-dual}. As a result,
    $$
    p^* \ge \sum_{i=1}^N \log \wtd\nu_i = \sum_{i=1}^N \log \hat\nu_i + N\log(N/\Gamma) = \hat p_\cG- N\log(\Gamma/N) \
    $$
    where the last equality is because $\hat \nu$ is the optimal solution of \eqref{discrete-dual}. 
    This completes the proof of \eqref{bound1}.  
    
    Finally, since $\hat \nu$ is a solution of \eqref{discrete-dual}, it holds $  \sum_{i=1}^N \hat \nu_i \varphi (X_i - \hat \mu_j)  \le N  $
    for all $j\in [M]$. Denote $X_{\min} = \min_{i\in [N]} X_i $ and $ X_{\max} =  \max_{i\in [N]} X_i$. 
    Since
    $\varphi$ is Lipschitz continuous, for any $   \mu \in [X_{\min}, X_{\max}] $, there exists $j\in [M]$ such that $ |\mu - \hat \mu_j| \le {\Dt \cG}/2$, hence 
    it holds
    \begin{eqnarray}
\Big| \sum_{i=1}^N \hat \nu_i \varphi (X_i - \mu) - \sum_{i=1}^N \hat \nu_i \varphi (X_i - \hat \mu_j) \Big| &\le&
\sum_{i=1}^N \hat \nu_i \frac{1}{\sqrt{2\pi e}} \frac{\Dt \cG}{2} \ , \nonumber
    \end{eqnarray}
    and thus
    $$
    \sum_{i=1}^N \hat \nu_i \varphi (X_i - \mu)  ~\le~ 
    \sum_{i=1}^N \hat \nu_i \varphi (X_i - \hat \mu_j) +  \frac{\Dt \cG}{\sqrt{8\pi e}}  \sum_{i=1}^N \hat \nu_i     ~\le~ 
    N+  \frac{\Dt \cG}{\sqrt{8\pi e}}  \sum_{i=1}^N \hat \nu_i \ . 
    $$
    On the other hand, it is easy to check that the function $\mu \mapsto  \sum_{i=1}^N \hat \nu_i \varphi (X_i - \mu) $ is increasing on $(-\infty, X_{\min}]$ and decreasing on $[X_{\max}, \infty)$, hence
    \begin{equation}\label{K-example}
    \max_{\mu \in \R}   \sum_{i=1}^N \hat \nu_i \varphi (X_i - \mu)    \le  N +  \frac{\Dt \cG}{\sqrt{8\pi e}}  \sum_{i=1}^N \hat \nu_i  \ . 
    \end{equation}
The above when combined with \eqref{bound1} leads to~\eqref{bound2}.

\subsubsection{Proof of Proposition~\ref{prop: error bound dual}}

	Recall that the function $F(\nu):= -\sum_{i=1}^N \log \nu_i$ is 1-self-concordant on $\R^N_{++}$. By \eqref{eqn: self-concordant strong convex} we have
	\begin{eqnarray}\label{ineq17}
	\sum_{i=1}^N (\log(\hat \nu_i) - \log(\nu_i^*) ) \ge
	\la \na F(\hat \nu), \nu^* - \hat \nu\ra + \rho \Big(\Big({\sum_{i=1}^N  {(\nu_i^* - \hat \nu_i)^2}/{\hat \nu_i^2}  }\Big)^{1/2}\Big) \ .
	\end{eqnarray}
	Let $\cS$ be the feasible region of \eqref{discrete-dual}. Since $\hat \nu$ is the optimal solution of \eqref{discrete-dual}, and $\nu^*$ is in $\cS$, it holds $\la \na F(\hat \nu), \nu^* - \hat \nu\ra \ge 0$.  As a result, \eqref{ineq17} implies
	\begin{eqnarray}
	\hat p_{\cG} - p^* ~=~ \sum_{i=1}^N \log(\hat \nu_i) - \log(\nu_i^*) ~\ge~   \rho \Big(\Big({\sum_{i=1}^N  {(\nu_i^* - \hat \nu_i)^2}/{\hat \nu_i^2}  }\Big)^{1/2}\Big) \ . \nonumber
	\end{eqnarray}
	The proof is complete by taking the inverse of $\rho(\cdot)$.

\section{Auxiliary results}\label{appendix section: Auxiliary results}

\begin{lemma}\label{lemma: Talor}
	For any third-order continuous differentiable function $f(\cdot)$ on a open set $Q\subseteq \R^M$ and any $w,y\in Q$, it holds
	\begin{eqnarray}
	&&f(y) - f(w) - \la \na f(w), y-w \ra - \frac{1}{2} \na^2 f(w ) [ y-w]^2  \nonumber\\
	&=&  \int_{0}^1  \int_{0}^1  \int_{0}^1  t^2 s \na^3 f\lt( w+ rst (y-w) \rt) [y-w]^3 \rmd t \ \rmd s\ \rmd r \nonumber
	\end{eqnarray} 
\end{lemma}

\begin{proof}
	Note that 
	\begin{eqnarray}
	f(y) - f(w)  - \la \na f(w), y-w \ra &=&
	\int_0^1 \lt[ \na f(w+t(y-w)) - \na f(w)  \rt]^\T  (y-w) \rmd t \nonumber\\
	&=&
	\int_0^1 \int_0^1  t \na^2 f(w+ ts (y-w)) [y-w]^2 \rmd s \ \rmd t \nonumber
	\end{eqnarray}
	Hence it holds that 
	\begin{eqnarray}
	&& f(y) - f(w)  - \la \na f(w), y-w \ra - \frac{1}{2} \na^2 f(w) [y-w]^2 \nonumber\\
	&=&
	\int_0^1 \int_0^1  t \Big(  \na^2 f(w+st(y-w)) - \na^2 f(w)  \Big) [y-w]^2 \rmd s \ \rmd t \nonumber\\
	&=&
	\int_0^1 \int_0^1  \int_0^1   t^2 s \na^3 f(w+str (y-w)) [y-w]^3 \rmd r \ \rmd s \ \rmd t \nonumber
	\end{eqnarray}
	~
\end{proof}

\begin{lemma}\label{lemma: matrix inequality}
	Let $\bar A\in \R^{n\times n}$ be a symmetric matrix and $\bar B \in \R^{n\times n} $ be a positive definite matrix (and hence invertible). Suppose $\bar A\preceq \bar B$ and $\bar A \succeq -\bar B$, then it holds $\bar A \bar B^{-1} \bar A \preceq \bar B$. 
\end{lemma}
\begin{proof}
	Let $P$ be a nonsingular matrix such that $P^\T \bar B P = I_n$. By assumption, we have 
	$$-I_n \preceq P^\T \bar A P \preceq I_n ~~
	\Rightarrow ~~~ (P^\T \bar A P)^2 \preceq I_n
	$$
	Therefore,
	$$
	P^\T \bar A \bar B^{-1} \bar A P = P^\T \bar A P \cdot (P^\T \bar BP)^{-1} \cdot  P^\T \bar A P =(P^\T \bar A P)^2 \preceq I_n
	$$
	which is equivalent to
	$$
	\bar A \bar B^{-1} \bar A \preceq P^{-\T} P^{-1} = \bar B .
	$$
	~
\end{proof}

\begin{lemma}\label{lemma: basic3}
    For any $x\in [0,1/5]$, it holds 
    \begin{equation}
    (1-x)^{-2} - 1 \le  3x \quad \text{and} \quad (1-x)^2 -1 \ge -3x.
    \end{equation}
\end{lemma}
\begin{proof}
    Note that the function $g(x) = (1-x)^{-2} - 1 - 3x $ is convex on $[0,1/5]$, with $g(0) = 0$ and $g(1/5) = 25/16 - 1 - 3/5= 9/16 - 9/15 <0$, so $g(x) \le 0$ for all $x\in[0,1/5] $. For the second inequality, note that $ (1-x)^2 - 1 = x^2 - 2x \ge -2x \ge -3x $ for all $x\ge 0$. 
\end{proof}

\section{Experimental Details}\label{sec:append:experimental-details}

For experiments in Section~\ref{sec:compute-unconstrained}, for Algorithm \ref{algorithm: Self-concordant Cubic-regularized Newton Method},
we use 
$ \ga_k = \rho_k=0.8^k $, and initialize with $ w^0 = (1/M) 1_M$. We set initially $L_0 = 3/\sqrt{2} $, and the enlargement parameter $ \beta=1.5$. 
For the first 10 iterations of Algorithm~\ref{algorithm: Self-concordant Cubic-regularized Newton Method}, we try a shorter step-size $ w^{k+1}  = w^k + 0.5(y^k - w^k)$. If this shorter step is accepted by \eqref{condition3}, then we use this shorter step. Otherwise, we use a full step $w^{k+1} = y^k$. After 10 iterations, 
we use the default full step $w^{k+1} = y^k$ for $k> 10$. 
 We terminate the Away-step Frank-Wolfe (AFW) method for subproblem~\eqref{update-yk} as follows: in the $k$-th outer iteration, let $h^{k,t}$ be the objective value of the subproblem \eqref{update-yk} in the  $t$-th iteration of Algorithm \ref{alg: AFW}. We terminate Algorithm \ref{alg: AFW} at iteration $t$ if $|h^{k,t-1}-h^{k,t}|/\max\{|h^{k,t-1}|,1\}< \text{tol}_{k} $ for some given $\text{tol}_{k}>0$. In particular, we set $\text{tol}_{k}=10^{-8} $ for $k\le 3$, $\text{tol}_{k}=10^{-9}$ for $4\le k\le 9$, and $\text{tol}_{k}=10^{-10} $ for $k\ge 10$. 

For experiments in Section~\ref{sec:compute-constrained}, we adopt the same parameter specifications mentioned above, except that we don't try the shorter steps ($w^{k+1}  = w^k + 0.5(y^k - w^k)$). 

For experiments in Section~\ref{sec:real-data}, we adopt the same parameter specifications as in Section~\ref{sec:compute-unconstrained}, except that we try $5$ iterations of shorter steps instead of $10$.

\end{document}